\documentclass[a4paper]{article}

\usepackage[ruled,vlined]{algorithm2e}
\usepackage{amsfonts,amsthm,amsmath,amssymb}
\usepackage{amstext}

\usepackage{subcaption}
\newcommand{\Diam}{\Delta}
\newcommand{\Rr}{\mathcal{R}}
\newcommand{\muR}{\mu_{\Rr}}
\newcommand{\Vorsepfam}{\mathcal{N}}

\newcommand{\Fopt}{\Ff_{\mathrm{OPT}}}
\newcommand{\Fheavy}{\Ff_{\mathrm{hv}}}
\newcommand{\epsloc}{\hat{\eps}}
\newcommand{\dmax}{d_{\max}}
\newcommand{\spn}{\mathrm{span}}

\newcommand{\td}[1]{\widetilde{#1}}
\newcommand{\tdFopt}{\widetilde{\Ff}_{\textrm{OPT}}}
\newcommand{\tdFheavy}{\widetilde{\Ff}_{\textrm{hv}}}

\newcommand{\tn}[1]{\tiny{#1}}
\newcommand{\sn}[1]{\scriptsize{#1}}

\newcommand{\enc}{\mathbf{enc}}
\newcommand{\exc}{\mathbf{exc}}

\newcommand{\Oh}{\mathcal{O}}
\newcommand{\eps}{\epsilon}

\newcommand{\poly}{\mathrm{poly}}

\newcommand{\dist}{\textrm{dist}}
\newcommand{\weight}{\mathbf{w}}

\newcommand{\Ss}{\mathcal{S}}
\newcommand{\Xfam}{\mathbb{X}}
\newcommand{\Xx}{\mathcal{X}}
\newcommand{\Ff}{\mathcal{F}}

\newcommand{\Cc}{\mathcal{C}}
\newcommand{\Dd}{\mathcal{D}}
\newcommand{\Pp}{\mathcal{P}}
\newcommand{\Uu}{\mathcal{U}}

\newcommand{\IntGraph}{\mathrm{IntGraph}}
\newcommand{\Ban}{\mathrm{Ban}}

\newcommand{\Exp}{\mathbb{E}}

\newcommand{\Prob}{\mathbb{P}}

\usepackage{listings}
\usepackage{enumitem}
  
\usepackage{graphicx,tikz}
\usepackage{url}
\usepackage{xspace}
\usepackage{geometry}
\usepackage{marginnote}

\def\cqedsymbol{\ifmmode$\lrcorner$\else{\unskip\nobreak\hfil
\penalty50\hskip1em\null\nobreak\hfil$\lrcorner$
\parfillskip=0pt\finalhyphendemerits=0\endgraf}\fi} 
\newcommand{\cqed}{\renewcommand{\qed}{\cqedsymbol}}
\newtheorem{theorem}{Theorem}

\newtheorem{lemma}[theorem]{Lemma}
\newtheorem{corollary}[theorem]{Corollary}
\newtheorem{claim}[theorem]{Claim}

\title{Quasi-polynomial time approximation schemes for packing and covering problems in planar graphs}

\author{Micha\l{} Pilipczuk\thanks{Institute of Informatics, University of Warsaw, Poland, michal.pilipczuk@mimuw.edu.pl. The research of Mi.\ Pilipczuk was partially supported by Polish National Science Centre grant UMO-2013/11/D/ST6/03073,
and was a part of projects that have received funding from the European Research Council (ERC) under the European Union's Horizon 2020 research and innovation programme (grant agreement No.~677651).}
\and Erik Jan van Leeuwen\thanks{Department of Information and Computing Sciences, Utrecht University, The Netherlands, e.j.vanleeuwen@uu.nl}
\and 
Andreas Wiese\thanks{Department of Industrial Engineering and Center for Mathematical Modeling, Universidad de Chile, Chile, awiese@dii.uchile.cl. The research of A. Wiese is supported by the Millennium Nucleus Information and Coordination in Networks ICM/FIC RC130003 and by the grant Fondecyt Regular 1170223.}}

\begin{document}

\maketitle

\begin{abstract}
We consider two optimization problems in planar graphs.
In {\sc{Maximum Weight Independent Set of Objects}} we are given a graph $G$ and a family $\mathcal{D}$ of {\em{objects}}, each being a connected subgraph of $G$ with a prescribed weight, 
and the task is to find a maximum-weight subfamily of $\mathcal{D}$ consisting of pairwise disjoint objects.
In {\sc{Minimum Weight Distance Set Cover}} we are given an edge-weighted graph $G$, two sets $\mathcal{D},\mathcal{C}$ of vertices of $G$, where vertices of $\mathcal{D}$ have prescribed weights, and a nonnegative radius $r$.
The task is to find a minimum-weight subset of $\mathcal{D}$ such that every vertex of $\mathcal{C}$ is at distance at most $r$ from some selected vertex.
Via simple reductions, these two problems generalize a number of geometric optimization tasks, notably {\sc{Maximum Weight Independent Set}} for polygons in the plane and 
{\sc{Weighted Geometric Set Cover}} for unit disks and unit squares.
We present {\em{quasi-polynomial time approximation schemes}} (QPTASs) for both of the above problems in planar graphs:
 given an accuracy parameter $\epsilon>0$ we can compute a solution whose weight is within multiplicative factor of $(1+\epsilon)$ 
from the optimum in time $2^{\mathrm{poly}(1/\epsilon,\log |\mathcal{D}|)}\cdot n^{\mathcal{O}(1)}$, where $n$
is the number of vertices of the input graph. 
Our main technical contribution is to transfer the techniques used for recursive approximation schemes for geometric problems 
due to Adamaszek, Har-Peled, and Wiese~\cite{AW2013,adamaszek2014qptas,Har-Peled2014} to the setting of planar graphs. In particular, this yields a purely
combinatorial viewpoint on these methods.
\end{abstract}
\reversemarginpar
\marginnote{\includegraphics[width=0.12\textwidth]{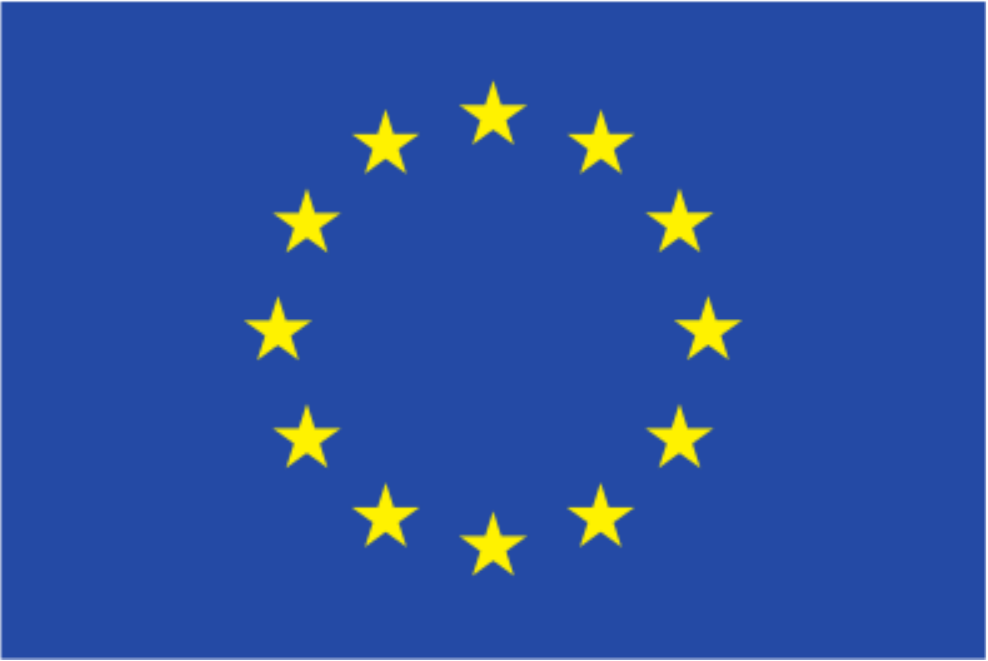}}[4.3cm]
\marginnote{\includegraphics[width=0.12\textwidth]{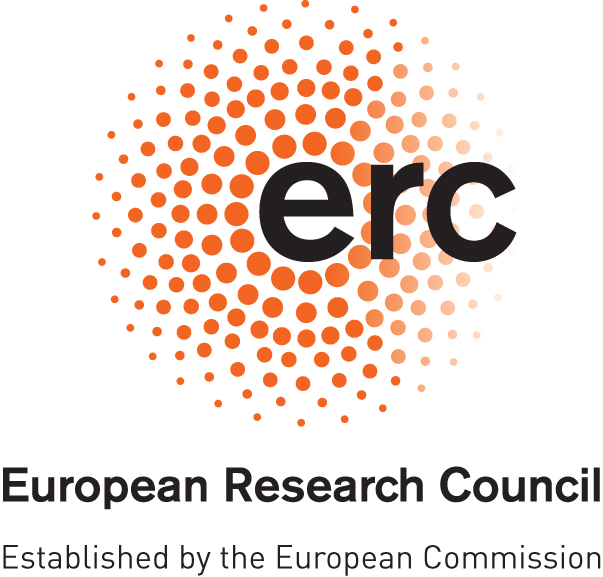}}[5.9cm]
\section{Introduction}\label{sec:intro}
\textsc{Independent Set} and \textsc{Dominating Set} are fundamental
optimization problems on graphs. Given a graph $G$ where each
vertex $v$ has a weight $\mathbf{w}(v)$, in \textsc{Independent Set} one seeks to find
a vertex subset $I\subseteq V(G)$ of maximum possible weight such that no two vertices
in $I$ are adjacent, whereas in \textsc{Dominating Set} one searches for a vertex subset $D$ of minimum possible weight such that each vertex $v\in V(G)$ is contained
in $D$ or adjacent to a vertex in~$D$. Even in the unit-weight setting, both problems are notoriously hard
to approximate and they are also $\mathsf{W}[1]$-hard, i.e., we do not expect that they admit {\em{fixed-parameter tractable}} ({\em{fpt}}) algorithms running in time $f(k)\cdot n^{\mathcal{O}(1)}$,
where $k$ is the expected solution size.

Therefore, special cases of the problems were investigated, for instance
the case where the input graph is planar. On planar graphs, classic layering techniques can be applied to show that both problems
admit EPTASs, i.e., $(1+\epsilon)$-approximation algorithms with
a running time of $f(1/\epsilon)n^{\mathcal{O}(1)}$ for some function $f$,
and fpt algorithms for the parameterization by the solution size, i.e., for a parameter $k$, algorithms
running in time $f(k)n^{\mathcal{O}(1)}$ that find a best solution among those of size at most $k$. 
Given these results, it is
natural to consider generalizations of the above problems on planar
graphs.

In this paper we study the {\sc{Distance Independent Set}} and
the {\sc{Distance Dominating Set}} problems. Given additionally a value
$r\in\mathbb{R}$ and weights on the edges of $G$, in the {\sc{Distance Independent Set}} problem we require that any two selected vertices in $I$ are
at distance larger than $r$ from each other, and in the {\sc{Distance Dominating Set}} problem we require that each vertex $v\in V(G)$ is at distance
at most $r$ from some vertex of $D$. Let us stress that we assume that $r$ is part of the input and in particular not assumed to be a constant;
in fact, for constant $r$ and unit edge weights, it is well-known that the same layering techniques easily yield EPTASs and fpt algorithms on planar graphs.
In the parameterized setting, both problems are $\mathsf{W}[1]$-hard even for unit weights; however, the trivial
$n^{\mathcal{O}(k)}$-time algorithms can be improved to $n^{\mathcal{O}(\sqrt{k})}$-time algorithms~\cite{MarxP15}.
These parameterized algorithms extend a technique originally developed to design quasi-polynomial time approximation schemes (QPTASs) for \textsc{Independent Set} and \textsc{Dominating Set} in the geometric (Euclidean) setting~\cite{AW2013,adamaszek2014qptas,Har-Peled2014}. The idea is to guess a sparse separator that has only small intersection with the optimal solution and that splits the problem into two roughly equal-sized subproblems, and then to solve the subproblems recursively. The natural question arises whether one can transfer the insights obtained in the parameterized setting back to approximation algorithms, and obtain approximation schemes for {\sc{Distance Independent Set}} and {\sc{Distance Dominating Set}} in planar graphs.

\subparagraph*{Our contribution.} In this paper
we show that this is indeed possible and we present the first quasi-polynomial time approximation schemes for {\sc{Distance Independent Set}} and {\sc{Distance Dominating Set}} on planar graphs when $r$ is part of the input.
In fact, we give QPTASs for two even more general problems, which we name \textsc{{Maximum Weight Independent Set of Objects}} ({\sc{MWISO}}) and \textsc{{Minimum
Weight Distance Set Cover}} ({\sc{MWDSC}}). In {\sc{MWISO}}
we are given a graph $G$ and a family $\mathcal{D}$ of {\em{objects}},
each being a connected subgraph of $G$ with a prescribed weight,
and the goal is to find a maximum-weight subfamily of $\mathcal{D}$ consisting
of pairwise disjoint objects. 
In {\sc{MWDSC}} we are given an edge-weighted graph $G$, subsets
of vertices $\mathcal{D}$ and $\mathcal{C}$ where vertices of $\mathcal{D}$ are weighted,
and radius $r\in\mathbb{R}$. The goal is to find a minimum-weight
subset $\mathcal{F}\subseteq\mathcal{D}$ that {\em{$r$-covers}} $\mathcal{C}$ in the
sense that each vertex of $\mathcal{C}$ is at distance at most $r$
from some vertex of $\mathcal{F}$. 
{\sc{MWISO}} generalizes {\sc{Distance Independent Set}} by taking $\mathcal{D}$ to be the family $\{\{v\colon \textrm{dist}(u,v)\leq r/2\}\colon u\in V(G)\}$ of all balls of radius $r/2$ in the graph,
while {\sc{MWDSC}} generalizes {\sc{Distance Dominating Set}} by taking $\mathcal{C}=V(G)$.
The following statements summarize our results.

\begin{theorem}\label{thm:main-mwis} The \textsc{{Maximum Weight
Independent Set of Objects}} problem in planar graphs admits a QPTAS
with running time $2^{\mathrm{poly}(1/\epsilon,\log N)}\cdot n^{\mathcal{O}(1)}$, where
$n$ is the vertex count of the input graph and $N=|\mathcal{D}|$ is the
number of objects in the input. 
\end{theorem}
\begin{theorem}\label{thm:main-mwds} The \textsc{{Minimum Weight
Distance Set Cover}} problem in planar graphs admits a QPTAS
with running time $2^{\mathrm{poly}(1/\epsilon,\log N)}\cdot n^{\mathcal{O}(1)}$, where
$n$ is the vertex count of the input graph and $N=|\mathcal{D}|$ is the
number of vertices allowed to be selected to the solution. 
\end{theorem}

To obtain our QPTASs for {\sc{MWISO}}, we extend the machinery
developed in \cite{AW2013,adamaszek2014qptas,Har-Peled2014} for optimization
problems in geometric settings to problems in planar graphs. The heart
of our technical contribution is to show that for any instance of
the above problems there is a set of candidate separators of polynomial size
such that one of them splits the given problem in a balanced way and
intersects only a tiny fraction of the given solution. The latter is important
since the intersected objects will be lost (in the case of {\sc{MWISO}}) or might be
paid twice (in the case of {\sc{MWDSC}}) and hence we need to bound their total weight
by $\epsilon \mathrm{OPT}$. We state here an informal version of our separator
lemma for the case of {\sc{MWISO}}.
\begin{lemma}[Informal]
\label{lem:breaking-informal} 
In polynomial time we can compute
a set $\mathbb{X}\subseteq2^{\mathcal{D}}$ of separators such that
for every solution $\mathcal{F}\subseteq \mathcal{D}$, say of weight $W$,
there exists $\mathcal{X}\in\mathbb{X}$ such that $\mathbf{w}(\mathcal{F}\cap\mathcal{X})\leq\epsilon W$
and in the intersection graph of $\mathcal{D}-\mathcal{X}$ each connected component
$\mathcal{C}$ satisfies $\mathbf{w}(\mathcal{C}\cap \mathcal{F})\le\frac{9}{10}W$. 
\end{lemma}
Using Lemma~\ref{lem:breaking-informal} as abstraction for finding
 separators, we can apply the same recursive scheme as~\cite{AW2013,adamaszek2014qptas,Har-Peled2014}:
we guess the correct separator $\mathcal{X}\in\mathbb{X}$, construct a subproblem for each connected component of the intersection graph of $\mathcal{D}-\mathcal{X}$, and recurse in
each of them up to recursion depth $\mathcal{O}(\log |\mathcal{D}|)$.
Thus, the only part of the reasoning that uses planarity is Lemma~\ref{lem:breaking-informal}.

The proof of Lemma~\ref{lem:breaking-informal} follows the reasoning
of Har-Peled~\cite{Har-Peled2014}. The idea is to prove the following
auxiliary result: for the optimal solution $\mathcal{F}$ (and in fact for any feasible solution) there exists a
separator of length roughly $s=\mathcal{O}(\frac{1}{\epsilon}\ln\frac{1}{\epsilon})$ that cuts through at most an $\epsilon$-fraction
of the weight of $\mathcal{F}$ and splits the weight of $\mathcal{F}$ in a balanced
way. Lemma~\ref{lem:breaking-informal} then follows by enumerating
all candidates for such separators. In~\cite{Har-Peled2014}, the
separator was simply a polygon with roughly $s$ vertices. We lift
this concept to planar graphs using {\em{Voronoi separators}}
as in the work of Marx and Pilipczuk~\cite{MarxP15}. Intuitively,
a Voronoi separator of length $r$ is an alternating cyclic sequence
of $r$ objects from $\mathcal{D}$ and $r$ faces of the graph, connected
by shortest paths in order to form a closed curve; this curve splits
the instance into two subinstances. Thus, shortest paths in the graph
are the analogues of segments in the plane.

The auxiliary result is proved in~\cite{Har-Peled2014} by showing
that if $\mathcal{S}$ is a sample of size roughly $s^{2}$ from $\mathcal{F}$, where
each object is sampled independently with probability proportional
to its weight, then a balanced separator of length $s$ in the Voronoi
diagram of $\mathcal{S}$ satisfies all the required properties with high
probability. We follow the same reasoning, however again we need to
properly understand how geometric concept used in~\cite{Har-Peled2014}
--- {\em{spokes}} and {\em{corridors}} --- should be interpreted
in planar graphs. Here, the technical toolbox for Voronoi diagrams
and Voronoi separators developed in~\cite{MarxP15} becomes very
useful. In particular, it turns out that a fine understanding of what
faces are candidates for branching points of a Voronoi diagram, provided
in~\cite{MarxP15}, is essential to make the probabilistic argument
work. Let us remark that we also somewhat simplify the original argument
of Har-Peled by replacing the Exponential Decay Lemma with a direct
probabilistic calculation. 

To give the QPTAS for {\sc{MWDSC}} we provide a variant of Lemma~\ref{lem:breaking-informal}
suitable for this problem and then follow a similar
recursive scheme as for Theorem~\ref{thm:main-mwis}. It is nice
that we can reuse the above-mentioned auxiliary result introduced for Lemma~\ref{lem:breaking-informal}
as a black-box, so the proof of the variant is relatively short. As
in~\cite{MustafaRR15}, the difference is that in the recursion instead
of removing the guessed separator we preserve it in all the recursive
subcalls, thus allowing double-buying objects from it. 

\subparagraph*{Geometric problems.} The above recursive machinery based on balanced separators was first
introduced for obtaining a QPTAS for \textsc{{Maximum Weight Independent
Set of Rectangles}} in the two-dimensional plane \cite{AW2013} and
then extended for getting QPTASs for \textsc{{Maximum Weight Independent
Set of Polygons}} \cite{adamaszek2014qptas,Har-Peled2014} and \textsc{{Weighted Geometric Set Cover}} ({\sc{WGSR}}) for pseudo-disks~\cite{MustafaRR15}.
We prove that Theorems~\ref{thm:main-mwis} and~\ref{thm:main-mwds} generalize these results, with the exception that for {\sc{WGSR}} we can treat only the cases of unit disks and axis-parallel unit squares, instead
of general families of pseudo-disks.
In Appendix~\ref{app:geometric} we explain how to derive the mentioned results from our theorems. 

We would like to comment that it is possible to reduce \textsc{MWISO} to \textsc{{Maximum Weight Independent
Set of Polygons}} and obtain a QPTAS for \textsc{MWISO} in this way: take the input graph and compute a 
straight line embedding for it. For each $p\in \mathcal{D}$ compute a spanning tree $T_p$ and define a polygon $P_p$ that consists of the edges
of $T_p$. Now two polygons $P_p, P_{p'}$ overlap if and only if their corresponding objects $p,p'$ overlap. Applying the QPTAS for 
\textsc{{Maximum Weight Independent Set of Polygons}} \cite{Har-Peled2014} to the resulting instance thus yields a QPTAS for \textsc{MWISO}.
However, we believe that our QPTAS for \textsc{MWISO} is simpler in the sense that it works with the planar graph directly.
Also note that such an approach does not work for {\sc{MWDSC}}.

\section{Proof of the Separator Lemma for {\sc{MWISO}}}\label{sec:voronoi}
In this section we prove the Separator Lemma for {\sc{MWISO}}, which was informally stated as Lemma~\ref{lem:breaking-informal} and is formally stated below.
For a family $\mathcal{D}$ of objects, $\mathrm{IntGraph}(\mathcal{D})$ denotes the {\em{intersection graph}} of $\mathcal{D}$: graph with vertex set $\mathcal{D}$ where two objects are adjacent iff they intersect.
The reader may think of $\mathcal{F}$ being the optimal solution and of $W$ being its weight.
\begin{lemma}[Separator Lemma for {\sc{MWISO}}]\label{lem:breaking}
Let $G$ be an $n$-vertex planar graph and $\mathcal{D}$ be a weighted family of $N$ objects in $G$.
Let $0<\epsilon<\frac{1}{10}$ and denote $s=10^3\cdot\frac{1}{\epsilon} \ln \frac{1}{\epsilon}$.
Then there exists a family $\mathbb{X}$ consisting of subsets of $\mathcal{D}$ with the following properties:
\begin{enumerate}[label=(A\arabic*),ref=(A\arabic*)]
\item\label{p:osl-size} $|\mathbb{X}|\leq 6^{3s}N^{15s}$ and $\mathbb{X}$ can be computed in time $N^{\mathcal{O}(s)}\cdot n^{\mathcal{O}(1)}$; and
\item\label{p:osl-sep}  for every real $W\geq 0$ and subfamily $\mathcal{F}\subseteq \mathcal{D}$ of pairwise disjoint objects such that $\mathbf{w}(\mathcal{F})\leq W$ and $\mathbf{w}(p)\leq s^{-2}W$ for each $p\in \mathcal{F}$,
there exist $\mathcal{X}\in \mathbb{X}$ such that $\mathbf{w}(\mathcal{F}\cap \mathcal{X})\leq \epsilon W$ and for every connected component $\mathcal{C}$ of $\mathrm{IntGraph}(\mathcal{D})\setminus\mathcal{X}$ we have $\mathbf{w}(\mathcal{C}\cap \mathcal{F})\leq \frac{9}{10}W$.
\end{enumerate}
\end{lemma}

The plan is as follows.
We first recall the toolbox of {\em{Voronoi separators}}, introduced by Marx and Pilipczuk~\cite{MarxP15}.
This allows us to state a stronger lemma, phrased as the existence of a short Voronoi separator appropriately breaking $\mathcal{F}$.
We then show how Lemma~\ref{lem:breaking} follows from this stronger result and subsequently prove the stronger result.

Before we proceed, let us set up the notation and basic assumptions about the input.
Let $G$ be the input graph.
We assume the edges of $G$ are assigned positive weights\footnote{Obviously, edge weights are immaterial for the {\sc{MWISO}} problem. However, it is convenient to think of $G$ as
edge-weighted so that we can define Voronoi diagrams. Furthermore, many results of this section will be reused for the {\sc{MWDSC}} problem, where edge weights play a role in the problem.}
so that we have the shortest-path metric in $G$: $\textrm{dist}(u,v)$ denotes the shortest length of a path connecting $u$ and $v$ in $G$.
We may assume that $G$ is given with an embedding in a sphere $\Sigma$ and that it is triangulated; that is, every face of $G$ is a triangle.
Indeed, adding edges of infinite weight to triangulate the graph neither distorts the metric nor spoils the connectivity of the objects.
Also, for every face $f$ of $G$ we fix any its internal point $c$ to be its {\em{center}}, and for each vertex $u$ of $f$ we fix some curve within $f$ with endpoints $u$ and $c$ 
to be the {\em{segment}} connecting $u$ and $c$
so that segments connecting vertices of $f$ with $c$ pairwise do not cross.

We assume that shortest paths are unique and the distances between pairs of vertices are pairwise different:
for every pair of vertices $u,v$ there is a unique shortest path connecting $u$ and $v$ and for $\{u,v\}\neq \{u',v'\}$ we have $\textrm{dist}(u,v)\neq \textrm{dist}(u',v')$. 
This can be ensured by using lexicographic tie-breaking rules and it increases the running time only by polynomial factors.

We are given a family $\mathcal{D}$ of {\em{objects}} in $G$, where each object $p\in \mathcal{D}$ is a nonempty, connected subgraph of $G$.
For any vertex $u$ of $G$ and any object $p\in \mathcal{D}$, let $\textrm{dist}(u,p)$ be the length of the shortest path connecting $u$ with any vertex of $p$.
For each object $p\in \mathcal{D}$, we fix any spanning tree $T(p)$ of $p$.
We also assume that the family $\mathcal{D}$ is weighted: every object $p\in \mathcal{D}$ is assigned a nonnegative real weight $\mathbf{w}(p)$. 
For $\mathcal{F}\subseteq \mathcal{D}$ we denote $\mathbf{w}(\mathcal{F})=\sum_{p\in \mathcal{F}} \mathbf{w}(p)$.

\subsection{Basic toolbox}

\subparagraph*{Voronoi partitions and diagrams.}
A subfamily $\mathcal{F}\subseteq \mathcal{D}$ is {\em{independent}}
if objects in $\mathcal{F}$ are pairwise vertex-disjoint.
Such an independent subfamily $\mathcal{F}$ induces the {\em{Voronoi partition}} $\mathbb{M}_\mathcal{F}$, which is a partition of the vertex set of $G$ into $|\mathcal{F}|$ parts according to the closest object from $\mathcal{F}$.
Precisely, for $p\in \mathcal{F}$, we say that a vertex $u$ belongs to the {\em{Voronoi cell}} $\mathbb{M}_\mathcal{F}(p)$ if $\textrm{dist}(u,p)<\textrm{dist}(u,p')$ for any $p'\in \mathcal{F}$, $p'\neq p$. Observe that ties do not happen due to distinctness
of distances in $G$.
We note that Marx and Pilipczuk consider in~\cite{MarxP15} a more general notion of a {\em{normal subfamily}}, but we do not need this generality here.

\begin{figure}[h!]
        \centering
        \begin{subfigure}[b]{0.46\columnwidth}
                \centering
                \def\svgwidth{0.98\columnwidth}
                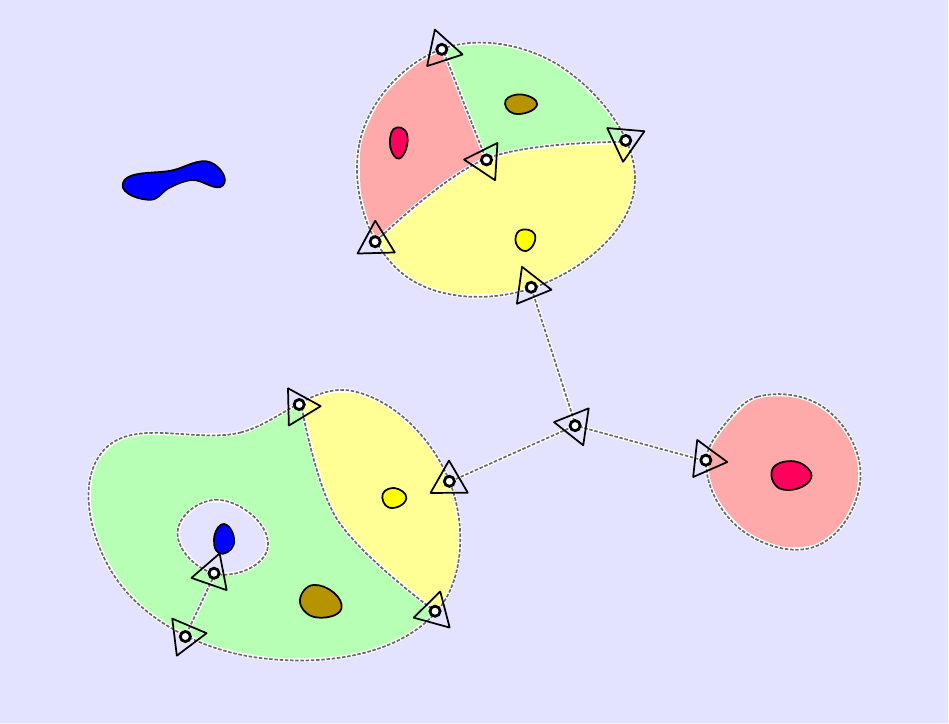
                \caption{Voronoi partition and diagram}
        \end{subfigure}
        \quad
        \begin{subfigure}[b]{0.46\columnwidth}
                \centering
                \def\svgwidth{0.98\columnwidth}
                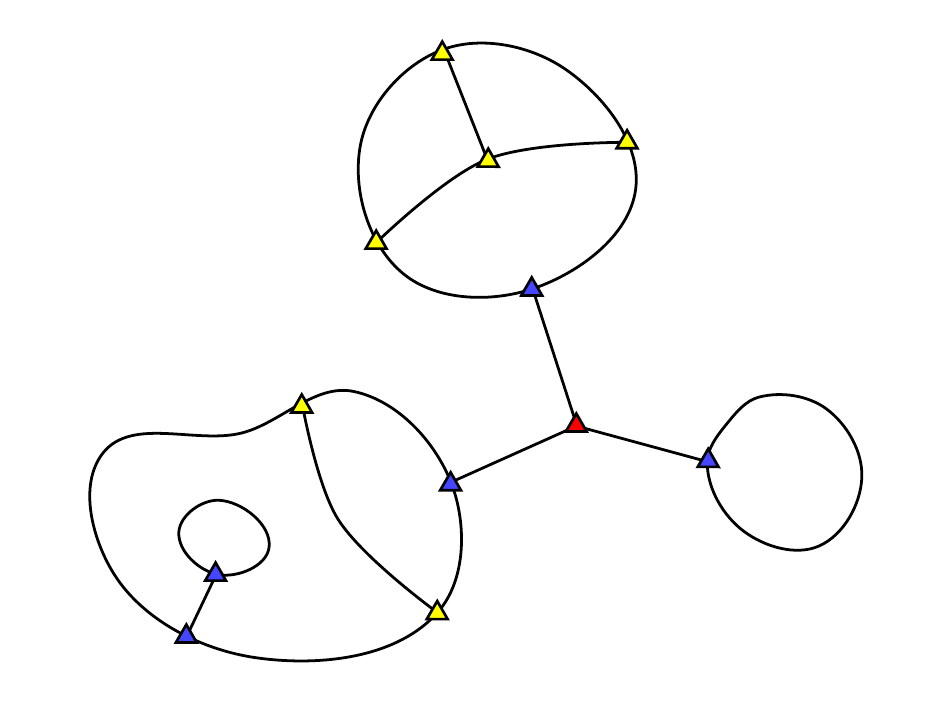
                \caption{Voronoi diagram, abstractly}
        \end{subfigure}
\caption{Construction of the Voronoi diagram. The left panel depicts the Voronoi partition, branching points of the diagram are depicted as solid (primal) triangular faces, edges of the diagram are dashed.
The right panel depicts the Voronoi diagram alone; note that it is a planar 3-regular multigraph. Branching points of types 1, 2, and 3 are depicted in yellow, blue, and red, respectively.
The figures are taken from~\cite{MarxP15} with consent of the authors; the left one is slightly modified.}\label{fig:diagram}
\end{figure}

Assuming $|\mathcal{F}|\geq 4$, we define the {\em{Voronoi diagram}} induced by $\mathbb{M}_\mathcal{F}$ as follows; see Figure~\ref{fig:diagram}
for a visualization.
First, observe that every Voronoi cell $\mathbb{M}_\mathcal{F}(p)$, for $p\in \mathcal{F}$, induces a connected subgraph of $G$ that contains $p$ entirely (see Lemmas 4.1 and 4.2 in~\cite{MarxP15}).
Extend $T(p)$ to a spanning tree $\widehat{T}(p)$ of $G[\mathbb{M}_\mathcal{F}(p)]$ by adding, for each vertex $u$ of $\mathbb{M}_\mathcal{F}(p)$ that is not in $p$, the shortest path from $u$ to $p$.
Take the dual of $G$ and remove all the edges dual to the edges of $\widehat{T}(p)$, for all $p\in \mathcal{F}$.
Then exhaustively remove vertices of degree one, and finally replace each maximal path with internal vertices of degree~$2$ (so-called {\em{$2$-path}})
by a single edge; the embedding of this edge is defined as the union of embeddings of edges of $G$ comprising the original $2$-path.
Thus, we obtain a connected, $3$-regular plane multigraph, called the {\em{Voronoi diagram}} of $\mathcal{F}$, whose faces bijectively correspond to the cells $\mathbb{M}_\mathcal{F}(p)$ for $p\in \mathcal{F}$.
More precisely, every face of the Voronoi diagram of $\mathcal{F}$ is associated with a different object $p\in \mathcal{F}$ so that all the vertices of $\mathbb{M}_\mathcal{F}(p)$ are contained in this face.
The $3$-regularity of the diagram follows from the assumption that $G$ is triangulated. From Euler's formula it follows that if $|\mathcal{F}|=k$, then $H$ has $k$ faces, $2k-4$ vertices, and $3k-6$ edges.
See Lemmas 4.4 and 4.5 of~\cite{MarxP15} for a formal verification of these assertions, and Section 4.4 of~\cite{MarxP15} for a detailed description of the~construction.

\subparagraph*{Branching points.} If $H$ is the Voronoi diagram of an independent subfamily $\mathcal{F}\subseteq \mathcal{D}$, then $H$ is constructed from a subgraph of the dual of $G$ by contracting maximal $2$-paths.
Hence, vertices of $H$ correspond to faces of $G$.
These primal faces, equivalently dual vertices, are called the {\em{branching points}} of the diagram $H$; intuitively, these are faces where the boundaries of Voronoi cells meet nontrivially.
A priori, every face of $G$ could be a branching point of the Voronoi diagram of some independent subfamily $\mathcal{F}\subseteq \mathcal{D}$. However, in~\cite{MarxP15} it is proved that the number of
candidates for branching points can be bounded polynomially in~$|\mathcal{D}|$.

\begin{theorem}[Theorem 4.7 of~\cite{MarxP15}]\label{thm:important-enumeration}
There exists a family $I$ of faces of $G$ with $|I|\leq |\mathcal{D}|^4$ such that the following holds: 
for every independent subfamily of objects $\mathcal{F}\subseteq \mathcal{D}$, all branching points of the Voronoi diagram of $\mathcal{F}$ are contained in $I$.
Moreover, $I$ can be computed in time polynomial in $|\mathcal{D}|$ and $n$.
\end{theorem}

We fix the family $I$ provided by Theorem~\ref{thm:important-enumeration} and call its members {\em{$\mathcal{D}$-important}} faces of $G$.

\subparagraph*{Voronoi separators.} We now recall the concept of {\em{Voronoi separators}}; see Figure~\ref{fig:vorsep} for a visualization.
A {\em{Voronoi separator}} is a sequence of the form $$S=\langle p_1,u_1,f_1,v_1,p_2,u_2,f_2,v_2,\ldots,p_r,u_r,f_r,v_r\rangle,$$
where $p_i$ are pairwise disjoint objects from $\mathcal{D}$, $f_i$ are faces of $G$, and $u_i,v_i$ are distinct vertices lying on the face $f_i$.
For each $i\in \{1,\ldots,r\}$, define $P_i$ to be the shortest path from $u_i$ to $p_i$ and $Q_i$ to be the shortest path from $v_i$ to $p_{i+1}$, where indices behave cyclically.
For a Voronoi separator $S$ as above, its {\em{length}} is $r$ and its {\em{set of traversed objects}} is $\mathcal{D}(S)=\{p_1,\ldots,p_r\}$.

In the notation above, an object $q\in \mathcal{D}$ is {\em{banned}} by the separator $S$ if either $q$ intersects some object $p\in \mathcal{D}(S)$, or there is a vertex $w$ on some path $P_i$ such that
$\textrm{dist}(w,q)<\textrm{dist}(w,p_i)$, or there is a vertex $w$ on some path $Q_i$ such that $\textrm{dist}(w,q)<\textrm{dist}(w,p_{i+1})$. In particular, $\mathcal{D}(S)$ is banned by $S$.
Let $\mathrm{Ban}(S)$ denote the set of objects in $\mathcal{D}$ banned by $S$. 
Intuitively, the banned objects are those that are intersected by the separator and are lost when we recurse (in MWISO) or that might be selected and paid twice (in MWDSC). 
Therefore, we will later ensure that their total weight is small.

The following result is the aforementioned key step toward the proof of Lemma~\ref{lem:breaking}. It may be regarded as a lift of Theorem 4.22 from~\cite{MarxP15} or of Lemma 4.1 from~\cite{Har-Peled2014}
to our setting.

\begin{lemma}\label{lem:balanced-sep}
Let $W$ be a positive real, $0<\epsilon<\frac{1}{10}$, and $s=10^3\cdot\frac{1}{\epsilon} \ln \frac{1}{\epsilon}$.
Suppose $\mathcal{F}\subseteq \mathcal{D}$ is an independent subfamily of objects such that $|\mathcal{F}|\geq 4$, $\mathbf{w}(\mathcal{F})\leq W$, and $\mathbf{w}(p)\leq s^{-2}W$ for all $p\in \mathcal{F}$.
Then there exist a Voronoi separator $S$ satisfying the following:
\begin{enumerate}[label=(B\arabic*),ref=(B\arabic*)]
\item\label{p:important} $\mathcal{D}(S)\subseteq \mathcal{F}$ and all faces traversed by $S$ are $\mathcal{D}$-important;
\item\label{p:length} the length of $S$ is at most $3s$;
\item\label{p:sacrifice} the total weight of objects of $\mathcal{F}$ banned by $S$ is at most $\epsilon W$;
\item\label{p:balanced} for every connected component $\mathcal{C}$ of $\mathrm{IntGraph}(\mathcal{D})-\mathrm{Ban}(S)$, we have $\mathbf{w}(\mathcal{C})\leq \frac{9}{10}W$.
\end{enumerate}
\end{lemma}

It is not hard to see that Lemma~\ref{lem:breaking} follows from Lemma~\ref{lem:balanced-sep}: we simply enumerate all candidates for a Voronoi separator $S$ satisfying \ref{p:important} and \ref{p:length},
a straightforward estimate using Theorem~\ref{thm:important-enumeration} shows that there are at most $6^{3s}N^{15s}$ of them, and for each candidate $S$ we add $\mathrm{Ban}(S)$ to the constructed family $\mathbb{X}$.

\begin{proof}[Proof of Lemma~\ref{lem:breaking} assuming Lemma~\ref{lem:balanced-sep}]
Let us inspect in how many ways a Voronoi separator $S$ satisfying properties~\ref{p:important} and \ref{p:length} can be selected. 
Since $S$ has length at most $3s$, there are at most $N^{3s}$ ways to select the sequence of consecutive objects visited by $S$.
Since every face traversed by $S$ is $\Dd$-important and there are at most $N^4$ $\Dd$-important faces in total, there are at most $N^{12s}$ ways to select the sequence of consecutive faces visited by $S$.
On each of these faces we need to pick a pair of different vertices as the entry and leaving point, giving $6$ ways per faces, so at most $6^{3s}$ in total. 
Thus, the separator $S$ may be selected in at most $6^{3s}N^{15s}$ different ways. 
Furthermore, since the set of $\Dd$-important faces can be computed in polynomial time by Theorem~\ref{thm:important-enumeration},
a family $\Vorsepfam$ of at most $6^{3s}N^{15s}$ Voronoi separators satisfying~\ref{p:important} and \ref{p:length} can be enumerated in time $N^{\Oh(s)}\cdot n^{\Oh(1)}$.

Now define the output family $\Xfam$ as follows:
$$\Xfam = \{\Ban(S)\colon S\in \Vorsepfam\}.$$
Clearly $\Xfam$ satisfies property~\ref{p:osl-size}, whereas that it also satisfies property~\ref{p:osl-sep} follows directly from Lemma~\ref{lem:balanced-sep}, 
properties~\ref{p:sacrifice} and~\ref{p:balanced}.
\end{proof}

Thus, we are left with proving Lemma~\ref{lem:balanced-sep}. The idea, borrowed from Har-Peled~\cite{Har-Peled2014}, is that we construct a sufficiently large random sample from $\mathcal{F}$, where the probability of
picking each object is proportional to its weight. Then we inspect the Voronoi diagram induced by the sample and we argue that with non-zero probability it has a short separator giving rise to the sought
Voronoi separator $S$. To implement this plan we need two ingredients: an appropriate lift of the sampling idea from~\cite{Har-Peled2014} and the analysis of how separators in the Voronoi diagram give rise to
Voronoi separators in the graph, which is essentially taken from~\cite{MarxP15} with some technical details added. These two ingredients are explained in the next two subsections.

\subsection{Sampling}

\subparagraph*{Spokes and diamonds.} We first adjust technical notions used by Har-Peled~\cite{Har-Peled2014} in the geometric context to our setting; 
see Figure~\ref{fig:diamonds} for their visualization.
Suppose we have an independent family of objects $\mathcal{F}$ and its Voronoi diagram $H=H_\mathcal{F}$.
Let $f$ be any branching point of $H$ and let $u$ be any vertex of $f$. Let $p\in \mathcal{F}$ be the object of $\mathcal{F}$ such that $u\in \mathbb{M}_\mathcal{F}(p)$.
The {\em{spoke}} of $u$ in $H$ is the shortest path from $u$ to $p$ in $G$; note all the vertices of this shortest path belong to the cell $\mathbb{M}_\mathcal{F}(p)$.

Consider any subfamily $\mathcal{S}\subseteq \mathcal{F}$ and let $H_\mathcal{S}$ be the Voronoi diagram induced by $\mathcal{S}$; in the following, we consider spokes in the diagram $H_\mathcal{S}$.
For any spoke $P$ in $H_\mathcal{S}$, say connecting a vertex $u$ with the object $p\in \mathcal{S}$ satisfying $u\in \mathbb{M}_\mathcal{S}(p)$, 
we say that $P$ is {\em{in conflict}} with an object $p'\in \mathcal{F}$ if there is a vertex $v$ on $P$ such that $\textrm{dist}(v,p')<\textrm{dist}(v,p)$.
Note that this implies $p'\notin \mathcal{S}$, because the spoke $P$ has to be entirely contained in $\mathbb{M}_\mathcal{S}(p)$.
Define the {\em{weight}} of $P$ with respect to $\mathcal{F}$
as the total weight of objects from $\mathcal{F}$ that are in conflict with $P$.

Further, suppose $e$ is an edge of $H_{\mathcal{S}}$, with endpoints $f_1,f_2$ (not necessarily different).
Let $p_1,p_2$ be the objects of $\mathcal{S}$ corresponding to the faces of $H_{\mathcal{S}}$ incident to $e$ (possibly $p_1=p_2$).
Let $u_{1,1},u_{1,2}$ be the vertices of $f_1$ such that $u_{1,1}\in \mathbb{M}_{\mathcal{S}}(p_1)$, $u_{1,2}\in \mathbb{M}_{\mathcal{S}}(p_2)$, and the edge $u_{1,1}u_{1,2}$ of $G$ crosses the edge $e$ of $H_{\mathcal{S}}$.
Similarly pick vertices $u_{2,1},u_{2,2}$ of $f_2$.
The {\em{diamond}} induced by $e$ is the closed curve $\Delta_{\mathcal{S}}(e)$ on $\Sigma$ formed by the union of: segments connecting the center of $f_1$ with $u_{1,1}$ and $u_{1,2}$, 
the unique path between $u_{1,2}$ and $u_{2,2}$ in $\widehat{T}(p_2)$, segments connecting the center of $f_2$ with $u_{2,2}$ and $u_{2,1}$, 
and the unique path between $u_{2,1}$ and $u_{1,1}$ in $\widehat{T}(p_1)$.
The {\em{interior}} of $\Delta_{\mathcal{S}}(e)$ is the unique region of $\Sigma\setminus \Delta_{\mathcal{S}}(e)$ that contains $e$.
The {\em{weight}} of $\Delta_{\mathcal{S}}(e)$ with respect to $\mathcal{F}$ is the total weight of objects of $\mathcal{F}$ that are entirely contained in the interior of $\Delta_\mathcal{S}(e)$; note that these objects do not belong to~$\mathcal{S}$.

We remark that while spokes were used in~\cite{Har-Peled2014} in the form roughly as above, diamonds correspond to the notion of a {\em{corridor}} from~\cite{Har-Peled2014}.

\subparagraph*{Sampling lemma: statement.} 
We now state the crucial technical result: there is a bounded-size subfamily of the optimum solution that induces 
a Voronoi diagram where every spoke and every diamond has small weight. We will prove it using a probabilistic sampling argument.

\begin{lemma}[Sampling lemma]\label{lem:sampling}
Suppose $W$ is a positive real and $\mathcal{F}$ is an independent, weighted family of objects in $G$ such that $|\mathcal{F}|\geq 4$ and $\mathbf{w}(\mathcal{F})\leq W$.
Let $\ell\geq 10$ be an integer such that $\mathbf{w}(p)\leq \frac{W}{\ell}$ for each $p\in \mathcal{F}$. Then
there exists a subfamily $\mathcal{S}\subseteq \mathcal{F}$ with $4\leq |\mathcal{S}|\leq 2\ell$ 
such that in the Voronoi diagram $H_\mathcal{S}$, the weight with respect to $\mathcal{F}$ of every spoke and of every diamond is at most $10\ln \ell \cdot \frac{W}{\ell}$.
\end{lemma}

Later, we will use Lemma~\ref{lem:sampling} with $\ell=\mathcal{O}((\frac{1}{\epsilon}\ln \frac{1}{\epsilon})^2)$. 
The reader may imagine that we then apply a balanced planar graph separator on the Voronoi diagram $H_\mathcal{S}$ of size $\mathcal{O}(\sqrt{\ell})$ along which we partition $\mathcal{F}$ into two parts, 
yielding the Voronoi separator claimed by Lemma~\ref{lem:balanced-sep}. 
Since the weight with respect to $\mathcal{F}$ of every spoke and every diamond is at most $10\ln \ell \cdot \frac{W}{\ell}$, the total weight of the objects of $\mathcal{F}$ banned by $S$ will be bounded by $\epsilon W$.

Lemma~\ref{lem:sampling} is a roughly an analogue of Lemma 3.3 from~\cite{Har-Peled2014}.
The main difference is that in the geometric setting, spokes and corridors have a simpler structure due to the fact that each branching point of the Voronoi diagram is defined by three objects from the solution ---
the three ones equidistant from it --- so that the branching point is the meeting point of the three corresponding Voronoi regions. This is no longer the case in planar graphs, as observed in~\cite{MarxP15}.
More precisely, out of the three regions around a branching point of the diagram, two or even three may be equal; this happens when there are bridges in the diagram, which is never the case in the geometric setting.

As part of their proof of Theorem~\ref{thm:important-enumeration}, to understand these additional situations Marx and Pilipczuk define {\em{singular faces}}, which come in three types.
The first one corresponds to ``standard'' branching points incident to three different regions, while the second and the third one correspond to branching points incident only to two, respectively one region.

\subparagraph*{Singular faces.}
For an independent triple of objects $\mathcal{F}_0=\{p_1,p_2,p_3\}\subseteq \mathcal{D}$, a face $f$ of $G$ is called a {\em{singular face}} of {\em{type $1$}} for $(p_1,p_2,p_3)$ if in $\mathbb{M}_{\mathcal{F}_0}$, 
all the vertices of $f$ belong to different cells (note that there are three cells in $\mathbb{M}_{\mathcal{F}_0}$).
For an independent triple of objects $\mathcal{F}_0=\{p_1,p_2,p_3\}\subseteq \mathcal{D}$, a face $f$ is called a {\em{singular face}} of {\em{type $2$}} for $(p_1,p_2,p_3)$ if in $\mathbb{M}_{\mathcal{F}_0}$, 
one vertex $v_1$ of $f$ belongs to $\mathbb{M}_{\mathcal{F}_0}(p_1)$, the other two vertices $v_2,v_3$ of $f$ belong to $\mathbb{M}_{\mathcal{F}_0}(p_2)$, and the closed walk $W$ obtained by taking the union of the unique path in
$\widehat{T}(p_2)$ between $v_2$ and $v_3$ and the edge $v_2v_3$ on the boundary of $f$ divides the plane into two regions, one containing $p_1$ and one containing $p_3$.
Finally, for an independent quadruple of objects $\mathcal{F}_0=\{p_0,p_1,p_2,p_3\}\subseteq \mathcal{D}$, a face $f$ is called a {\em{singular face}} of {\em{type $3$}} for $(p_0,p_1,p_2,p_3)$ if in $\mathbb{M}_{\mathcal{F}_0}$ all
the vertices of $f$ belong to $\mathbb{M}_{\mathcal{F}_0}(p_0)$, but the boundary of face $f$ plus the minimal subtree of $\widehat{T}(p_0)$ spanning the vertices of $f$ divides the plane into four regions: the face $f$
itself, one region containing $p_1$, one region containing $p_2$, and one region containing $p_3$.
See Figure~8 in~\cite{MarxP15} for a visualization.

It appears that for a fixed triple or quadruple of objects, there are only few singular faces.

\begin{lemma}[Lemmas 4.8, 4.9, and 4.10 of~\cite{MarxP15}]\label{lem:sing-bound}
For each independent triple of objects $(p_1,p_2,p_3)$, there are at most $2$ singular faces of type $1$ for $(p_1,p_2,p_3)$, and at most $1$ singular face of type $2$ for $(p_1,p_2,p_3)$.
For each independent quadruple of objects $(p_0,p_1,p_2,p_3)$, there is at most $1$ singular face of type $3$ for $(p_0,p_1,p_2,p_3)$.
\end{lemma}

The next statement explains the connection between branching points and singular faces.

\begin{lemma}[Lemma 4.12 of~\cite{MarxP15}]\label{lem:diag-sing}
Let $\mathcal{F}\subseteq \mathcal{D}$ be an independent subfamily of objects, and let $H$ be the Voronoi diagram of $\mathcal{F}$.
Then every branching point of $H$ is either a type-$1$ singular face for some triple of objects from $\mathcal{F}$, or a type-$2$ singular face for some triple of objects from $\mathcal{F}$, or a type-$3$ singular
face for some quadruple of objects from $\mathcal{F}$.
\end{lemma}

Actually, the two results above may be combined into a proof of Theorem~\ref{thm:important-enumeration}. Lemma~\ref{lem:diag-sing} shows that every branching point of the Voronoi diagram of an independent
subfamily of $\mathcal{D}$ is among type-$1$, type-$2$, and type-$3$ singular faces for triples or quadruples of objects in $\mathcal{D}$, while using Lemma~\ref{lem:sing-bound} we can bound their total number by $|\mathcal{D}|^4$.

\subparagraph*{Sampling lemma: proof.} We now have all the tools needed to prove Lemma~\ref{lem:sampling}. 
Contrary to Har-Peled~\cite{Har-Peled2014} we do not use the Exponential Decay Lemma, but direct probability calculations; this makes the proof somewhat conceptually easier.
The main complications are due to the need to handling different types of singular faces, instead of just one.

\begin{proof}[Proof of Lemma~\ref{lem:sampling}]
Denote $\eta=10\ln \ell \cdot \frac{W}{\ell}$.
First observe that if $\mathbf{w}(\mathcal{F})\leq \eta$, then setting $\mathcal{S}=\mathcal{F}$ satisfies all the required properties, since no spoke may have larger weight than the whole of $\mathcal{F}$.
Therefore, from now on we assume that $\mathbf{w}(\mathcal{F})>\eta$.

Construct $\mathcal{S}$ by including every object $p\in \mathcal{F}$ independently with probability $q_p=\mathbf{w}(p)\cdot \frac{\ell}{W}$; note that this value is at most $1$ by the assumption of the lemma.
Let $X$ be the random variable equal to the cardinality of $\mathcal{S}$; then $X=\sum_{p\in \mathcal{F}} X_p$, where $X_p$ are indicator random variables, taking value $1$ if $p$ is included in $\mathcal{F}$ and $0$ otherwise.
Note that $\mathbb{E}[X_p] = q_p$ and $\mathbb{E}[X] = \sum_{p\in \mathcal{F}} \mathbb{E} [X_p]=\ell\cdot \frac{\mathbf{w}(\mathcal{F})}{W}\leq \ell$.
Since $X$ is a sum of independent indicator variables, standard concentration inequalities yield the following.

\begin{claim}[
]\label{cl:cardinality}
The probability that $|\mathcal{S}|>2\ell$ or $|\mathcal{S}|\leq 100$ is at most $\frac{7}{12}$.
\end{claim}
\begin{proof}[Proof of claim]
Since $\Exp[X]\leq \ell$, by Markov's inequality we have that $|\Ss|>2\ell$ happens with probability at most $\frac{1}{2}$.
Hence,  it suffices to prove that the probability that $|\Ss|\leq 100$ is that most $\frac{1}{12}$.

As $\weight(\Ff)>\eta=10\ln \ell\cdot \frac{W}{\ell}$, we have
$$\Exp[X] = \ell\cdot \frac{\weight(\Ff)}{W}>10\ln \ell>200.$$
Since $X$ is a sum of independent $\{0,1\}$-random variables with mean larger than $200$, by Chernoff's inequality we infer that
$$\Prob(X\leq 100)\leq \exp\left(-\frac{200}{8}\right)\leq \frac{1}{12}.$$
This concludes the proof of the claim.
\cqed\end{proof}

Call a spoke in the Voronoi diagram $H_\mathcal{S}$ {\em{heavy}} if its weight with respect to $\mathcal{F}$ is more than $\eta$.
We now estimate the probability that there is a heavy spoke in~$H_\mathcal{S}$.

\begin{claim}\label{cl:spokes}
The probability that there is a heavy spoke in $H_\mathcal{S}$ is at most $\frac{1}{6}$.
\end{claim}
\begin{proof}[Proof of claim]
Take any triple $p_1,p_2,p_3$ of different objects from $\Ff$, and suppose $f$ is a singular face of type $1$ for $(p_1,p_2,p_3)$.
Suppose further that $P$ is one of the spokes incident to $f$ in the Voronoi diagram of $\{p_1,p_2,p_3\}\subseteq \Ff$, and assume that $P$ is heavy with respect to $\Ff$.
Let us estimate the probability of the following event $A$: $p_1,p_2,p_3$ all belong to $\Ss$ and moreover $P$ remains a spoke in the Voronoi diagram $H_\Ss$ of $\Ss$ (then it is obviously a heavy spoke in $H_\Ss$ as
well).
For $A$ to happen, apart from the event that $p_1,p_2,p_3\in \Ss$, 
we also need that the following event happens: for each $p'\in \Ff\setminus \{p_1,p_2,p_3\}$ that is in conflict with $P$, $p'$ is not included in $\Ss$.
Let $\mathcal{Z}\subseteq \Ff$ be the set of such objects $p'$. Using the standard inequality $1-x\leq \exp(-x)$, we have
\begin{eqnarray}
\Prob(A)& \leq & q_{p_1} q_{p_2} q_{p_3}\cdot \prod_{p'\in \mathcal{Z}} (1-q_{p'}) \leq q_{p_1} q_{p_2} q_{p_3}\cdot \exp\left(-\sum_{p'\in \mathcal{Z}} q_{p'}\right)\nonumber\\
& \leq & \frac{\weight(p_1)\weight(p_2)\weight(p_3)}{W^3}\cdot \ell^{3}\cdot \exp\left(-\frac{\ell}{W} \sum_{p'\in \mathcal{Z}} \weight(p')\right) \nonumber\\
& \leq & \frac{\weight(p_1)\weight(p_2)\weight(p_3)}{W^3}\cdot \ell^{3}\cdot \exp\left(-\frac{\ell\eta}{W}\right)\nonumber \\
& \leq & \frac{\weight(p_1)\weight(p_2)\weight(p_3)}{W^3}\cdot \ell^{3}\cdot \exp(-10\ln \ell) = \frac{\weight(p_1)\weight(p_2)\weight(p_3)}{W^3}\cdot \ell^{-7}\nonumber\\
& \leq & 10^{-7}\cdot \frac{\weight(p_1)\weight(p_2)\weight(p_3)}{W^3}.\label{eq:bound1}
\end{eqnarray}
For a fixed triple $p_1,p_2,p_3\in \Ff$, face $f$ that is singular of type $1$ for $(p_1,p_2,p_3)$, and spoke $P$ incident to $f$, 
let us denote by $A^1_{P,(p_1,p_2,p_3)}$ the event $A$ considered above.
Let $B^1$ denote the event that some event $A^1_{P,(p_1,p_2,p_3)}$ happens, i.e., $B^1$ is the union of all events $A^1_{P,(p_1,p_2,p_3)}$.
By Lemma~\ref{lem:sing-bound}, for every triple $p_1,p_2,p_3\in \Ff$ there are at most two faces $f$ that are singular of type $1$ for $(p_1,p_2,p_3)$, and for each of them there are at most $3$ spokes $P$ to consider.
Thus, by applying the union bound to~\eqref{eq:bound1} we infer that
$$\Prob(B^1)\leq 6\cdot 10^{-7}\cdot \sum_{p_1\in \Ff}\sum_{p_2\in \Ff}\sum_{p_3\in \Ff} \frac{\weight(p_1)\weight(p_2)\weight(p_3)}{W^3}=6\cdot 10^{-7}\cdot \frac{\weight(\Ff)^3}{W^3}\leq 10^{-6}.$$

Denote by $B^2$ the following event: there exists distinct $p_1,p_2,p_3\in \Ff$, a singular face $f$ of type $2$ for $(p_1,p_2,p_3)$ such that (i) $p_1,p_2,p_3\in \Ss$ and (ii) at least one of the heavy (w.r.t. $\Ff$)
spokes incident to $f$ in the Voronoi diagram of $\{p_1,p_2,p_3\}$ remains a spoke in $\Ss$.
Also, denote by $B^3$ the following event: there exists distinct $p_0,p_1,p_2,p_3\in \Ff$, a singular face $f$ of type $3$ for $(p_0,p_1,p_2,p_3)$ 
such that (i) $p_0,p_1,p_2,p_3\in \Ss$ and (ii) at least one of the heavy (w.r.t. $\Ff$) spokes incident to $f$ in the Voronoi diagram of $\{p_0,p_1,p_2,p_3\}$ remains a spoke in $\Ss$.
Similar calculations as for $B^1$ yield the following:
$$\Prob(B^2)\leq 10^{-6}\qquad \textrm{and}\qquad \Prob(B^3)\leq 10^{-6}.$$
Here, when estimating $\Prob(B^3)$ we sum over quadruples of objects instead of triples.
By Lemma~\ref{lem:diag-sing}, we conclude that the probability that any spoke in the diagram $H_\Ss$ is heavy is at most $3\cdot 10^{-6}\leq \frac{1}{6}$.
\cqed\end{proof}

We are left with diamonds.
Call a diamond $\Delta_\mathcal{S}(e)$ in $H_{\mathcal{S}}$ {\em{heavy}} if its weight with respect to $\mathcal{F}$ is larger than $\eta$.
The next check follows by essentially the same estimation as Claim~\ref{cl:spokes}.

\begin{claim}\label{cl:diamonds}
The probability that there is a heavy diamond in $H_\mathcal{S}$ is at most $\frac{1}{6}$.
\end{claim}
\begin{proof}[Proof of claim]
Suppose $e$ is an edge of $H_\Ss$ with endpoints $f_1,f_2$, and let $p_1,p_2$ be the objects of $\Ss$ corresponding to faces of $H_\Ss$ incident to $e$. 
By Lemma~\ref{lem:diag-sing}, each $f_t$ ($t\in \{1,2\}$) is either a type-$1$ singular face for a triple of objects from $\Ss$, 
or a type-$2$ singular face for a triple of objects from $\Ss$,
or a type-$3$ singular face for a quadruple of objects from $\Ss$. Suppose for a moment that both $f_1$ and $f_2$ are type-$1$ singular faces for triples of objects from $\Ss$; then it is easy to see that
$f_1$ is a type-$1$ singular face for $(p_1,p_2,q_1)$ for some object $q_1\in \Ss$, while $f_2$ is a type-$1$ singular face for $(p_1,p_2,q_2)$ for some object $q_2\in \Ss$.
We will further assume that $q_1\neq q_2$ and discuss the other case, as well as the cases when $f_1$ or $f_2$ are not type-$1$ singular faces, at the end, as the reasoning for them is analogous.

All in all, we have a quadruple of pairwise different objects $p_1,p_2,q_1,q_2$. For the diamond $\Diam=\Diam_\Ss(e)$ to arise in $H_{\Ss}$, the following two events need to happen simultaneously:
\begin{itemize}
\item objects $p_1,p_2,q_1,q_2$ need to simultaneously be included in $\Ss$; and
\item all objects entirely contained in the interior of $\Diam$ need to be not included in $\Ss$.
\end{itemize}
These two events are independent. The probability of the first is 
$$\ell^4\cdot \frac{\weight(p_1)\weight(p_2)\weight(q_1)\weight(q_2)}{W^4}.$$
Let $\cal G$ be the family of objects entirely contained in the interior of $\Diam$. By the assumption that $\Diam$ is heavy, we have that $\sum_{r\in \cal G} \weight(r)\geq \eta$.
The probability that no object of $\cal G$ is included in $\Ss$ is upper bounded by
\begin{eqnarray*}
\prod_{r\in \cal G} (1-q_r) & \leq & \prod_{r\in \cal G} \exp(-q_r)=\exp\left(-\sum_{r\in \cal G} q_r\right)\\
& = & \exp\left(-\frac{\ell}{W}\cdot \sum_{r\in \cal G} \weight(r)\right)\leq \exp\left(-10\ln \ell\right)=\ell^{-10}.
\end{eqnarray*}
Therefore, the probability that $\Diam$ arises in $H_{\Ss}$ is upper bounded by 
$$\frac{\weight(p_1)\weight(p_2)\weight(q_1)\weight(q_2)}{W^4}\cdot \ell^4\cdot \ell^{-10}=\frac{\weight(p_1)\weight(p_2)\weight(q_1)\weight(q_2)}{W^4}\cdot \ell^{-6}.$$
Now for every quadruple $(p_1,p_2,q_1,q_2)$ of objects from $\Ff$ there are at most $4$ diamonds induced by this quadruple as above, as there are at most $2$ type-$1$ singular faces for $(p_1,p_2,q_1)$, and likewise
for $(p_1,p_2,q_2)$. Therefore, the total probability that there exists a heavy diamond with both endpoints being singular faces of type-$1$ and $q_1\neq q_2$ is bounded by
$$4\cdot \ell^{-6}\cdot \sum_{(p_1,p_2,q_1,q_2)\in \Ff^4} \frac{\weight(p_1)\weight(p_2)\weight(q_1)\weight(q_2)}{W^4} =4\cdot \ell^{-6}\cdot \frac{\weight(\Ff)^4}{W^4}\leq 4\cdot 10^{-6},$$
where the last inequality follows from the assumption that $\ell\geq 10$.
The reasoning for the case when $q_1=q_2$ is analogous (we sum over triples instead of quadruples), and similarly for the cases when we are dealing with singular faces of type different than $1$.
The number of different cases one needs to consider is bounded by $100$, so summing all the probabilities we conclude that the probability that there is a heavy diamond in $H_\Ss$ is at most $\frac{1}{6}$.
\cqed\end{proof}

Concluding, assertion $4\leq |\mathcal{S}|\leq 2\ell$ does not hold with probability at most $\frac{7}{12}$,
there is a heavy spoke in $H_\mathcal{S}$ with probability at most $\frac{1}{6}$, 
and there is a heavy diamond in $H_\mathcal{S}$ with probability at most $\frac{1}{6}$.
Hence, with probability at least $\frac{1}{12}$ neither of the above holds, so there exists a subfamily $\mathcal{S}$ satisfying all the postulated conditions.
\end{proof}

\subsection{Balanced nooses}

We proceed with the proof of Lemma~\ref{lem:balanced-sep} by explaining the second ingredient: balanced separators in Voronoi diagrams.
In general, short embedding-respecting separators in the Voronoi diagram --- so-called {\em{nooses}} --- correspond to Voronoi separators we are looking for.
We start by defining nooses and showing how the existence of a {\em{sphere-cut decomposition}} of small width --- a hierarchical decomposition of the diagram using nooses --- 
implies the existence of a short noose that breaks the instance in a balanced way. 

We remark that in~\cite{Har-Peled2014}, this part of the reasoning is essentially done by considering the {\em{radial}} graph of the Voronoi diagram and applying the weighted balanced separator theorem of
Miller~\cite{Miller86} to it. Such approach would be also applicable in our case, but we find the approach via sphere-cut decompositions more explanatory regarding how separators in the (radial graph of the)
diagram correspond to separators in the instance.

\subparagraph*{Sphere-cut decompositions.} We now recall the framework of {\em{sphere-cut decompositions}}, which are embedding-respecting hierarchical decompositions of planar graphs.

A {\em{branch decomposition}} of a graph $H$ is a pair $(T,\eta)$ where $T$ is a tree with all internal nodes
having degree $3$, and $\eta$ is a bijection between the edge set of $H$ and the leaf set of $T$ (for clarity, we always use the term {\em{node}} for a vertex of a decomposition tree).
Take any edge $e$ of $T$ and consider removing it from $T$; then $T$ breaks into two subtrees, say $T_1,T_2$. Let $A_1,A_2$ be the preimages of the leaf sets of $T_1,T_2$ under $\eta$, respectively;
then $(A_1,A_2)$ is a partition of the edge set of $H$.
The {\em{width}} of the edge $e$ is the number of vertices of $H$ incident to both an edge of $A_1$ and to an edge of $A_2$, and the {\em{width}}
of the branch decomposition $(T,\eta)$ is the maximum among the widths of the edges of $T$. The {\em{branchwidth}} of $H$ is the minimum possible width of a branch decomposition of $H$.

Let $H$ be a connected plane graph embedded in a sphere $\Sigma$. A {\em{noose}} in $H$ is a closed, directed curve $\gamma$ on $\Sigma$ without self-crossings
that meets $H$ only at its vertices and visits every face of $H$ at most once.
Note that removing $\gamma$ from the sphere $\Sigma$ breaks it into two open disks: for one of them $\gamma$ is the clockwise traversal of the perimeter, and for the other it is the counter-clockwise traversal (fixing
an orientation of $\Sigma$). The first disk shall be called $\mathbf{enc}(\gamma)$ while the second shall be called $\mathbf{exc}(\gamma)$ (for {\em{enclosed}} and {\em{excluded}}).
Two nooses $\gamma,\gamma'$ are {\em{equivalent}} if they are homotopic on $\Sigma$ with a homotopy that fixes the embedding of $H$; in other words, $\gamma'$ can be obtained from $\gamma$ by continuous
transformations within faces of $H$.

A {\em{sphere-cut decomposition}} of $H$ is a triple $(T,\eta,\delta)$ where $(T,\eta)$ is a branch decomposition of $H$ and $\delta$ maps ordered pairs of adjacent nodes of
$T$ to nooses on $\Sigma$ (w.r.t.~$H$) such that the following conditions are satisfied for each pair of adjacent nodes of $T$:
\begin{itemize}
\item $\delta(x,y)$ is equal to $\delta(y,x)$ reversed (in particular $\mathbf{enc}(\delta(x,y))=\mathbf{exc}(\delta(y,x))$);
\item $\mathbf{enc}(\delta(x,y))$ contains all the edges of $H$ mapped to the component of $T-xy$ containing~$y$, while $\mathbf{exc}(\delta(y,x))$ contains all the edges of $H$ mapped to the other component of $T-xy$.
\end{itemize}
The following result follows from~\cite{GuT12,SeymourT94} and was formulated in exactly this way in~\cite{MarxP15}.

\begin{theorem}\label{thm:sc-decomp}
Every $n$-vertex sphere-embedded multigraph that is connected and bridgeless has a sphere-cut decomposition of width at most $\sqrt{4.5 n}$.
\end{theorem}

We note that in Theorem~\ref{thm:sc-decomp}, the assumption that the multigraph is bridgeless is necessary, as multigraphs with bridges do not have sphere-cut decompositions at all.

Suppose $(T,\eta,\delta)$ is a sphere-cut decomposition of $G$. It is straightforward to see that we may adjust the nooses $\delta(x,y)$ for $x,y$ ranging over adjacent nodes of $T$ so that the following
condition is satisfied: if node $x$ has neighbors $y_1,y_2,y_3$ in $T$, then 
$\mathbf{enc}(\delta(y_1,x))$ is the disjoint union of $\mathbf{enc}(\delta(x,y_2))$, $\mathbf{enc}(\delta(x,y_3))$, and $(\delta(x,y_2)\cap \delta(x,y_3))\setminus \delta(y_1,x)$.
Sphere-cut decompositions satisfying this condition will be called {\em{faithful}}.
It is easy to see that any sphere-cut decomposition can be made faithful by changing each noose to an equivalent one.

\subparagraph*{Separator theorem for nooses.}
We now state a separator theorem for nooses drawn from a sphere-cut decomposition of a given sphere-embedded multigraph.
The theorem is weighted with respect to a measure defined as follows. 
Suppose $\mathcal{R}$ is a finite family of pairwise disjoint {\em{objects}} on a sphere $\Sigma$,
where each object $p\in \mathcal{R}$ is a nonempty arc-connected subset of $\Sigma$ with associated nonnegative weight $\mathbf{w}(p)$.
For an open disk $\Delta\subseteq \Sigma$, define its {\em{$\mathcal{R}$-measure}} $\mu_{\mathcal{R}}(\Delta)$ as the total weight of objects from $\mathcal{R}$ that are entirely contained in $\Delta$.

\begin{lemma}\label{lem:sc-to-balanced}
Let $H$ be a connected, bridgeless multigraph embedded on a sphere $\Sigma$. 
Let $\mathcal{R}$ be a weighted family of pairwise disjoint objects on $\Sigma$ and let $W=\mathbf{w}(\mathcal{R})$.
Suppose further $(T,\eta,\delta)$ is a faithful sphere-cut decomposition of $H$ such that for every pair $(x,y)$ of adjacent nodes in $T$ such that $x$ is a leaf, we have $\mu_{\mathcal{R}}(\mathbf{enc}(\delta(y,x)))<\frac{9}{20}W$.
Then there exists a noose $\gamma$ w.r.t $H$, which is one of the nooses in the sphere-cut decomposition $(T,\eta,\delta)$, such that the following hold:
$$\mu_{\mathcal{R}}(\mathbf{enc}(\gamma))\leq \frac{9}{10}W\qquad\textrm{and}\qquad \mu_{\mathcal{R}}(\mathbf{exc}(\gamma))\leq \frac{9}{10}W.$$
\end{lemma}

The proof of Lemma~\ref{lem:sc-to-balanced} is standard: we find a balanced edge $(x,y)$ in the decomposition $(T,\eta,\delta)$ and $\delta(x,y)$ is the sought noose. The fact that nooses appearing in $(T,\eta,\delta)$ may intersect objects from $\mathcal{R}$ requires some technical attention.
We start by proving a general-use statement using which we can extract balanced separators from hierarchical decompositions, provided we want balanceness with respect to some measure $\mu(\cdot)$ that the
decomposition roughly obeys. 

\begin{lemma}\label{lem:balanced-edge}
Let $T$ be a tree with all internal nodes of degree $3$ and let $W$ be a positive real.
Suppose that there is a function $\mu$ which maps ordered pairs of adjacent nodes of $T$ to nonnegative reals such that the following conditions are satisfied:
\begin{enumerate}[label=(S\arabic*),ref=(S\arabic*)]
\item\label{p:edge} For every pair $x,y$ of adjacent nodes in $T$, we have $$\frac{9}{10}W\leq \mu(x,y)+\mu(y,x)\leq W.$$
\item\label{p:node} If $x$ is an internal node of $T$ with neighbors $y_1,y_2,y_3$, then $$\mu(x,y_2)+\mu(x,y_3)\leq \mu(y_1,x)\leq \mu(x,y_2)+\mu(x,y_3)+\frac{1}{10}W.$$
\item\label{p:leaf} Whenever $x,y$ are adjacent nodes in $T$ and $x$ is a leaf, we have $\mu(y,x)<\frac{9}{20}W$.
\end{enumerate}
Then there exists an edge $xy$ of $T$ such that
$$\frac{1}{4}W<\mu(x,y)<\frac{3}{4}W\qquad\textrm{and}\qquad\frac{1}{4}W<\mu(y,x)<\frac{3}{4}W.$$
\end{lemma}
\begin{proof}
Orient the edges of $T$ as follows: if $x,y$ are adjacent nodes, then orient $xy$ toward $x$ if $\mu(x,y)>\mu(y,x)$, toward $y$ if $\mu(x,y)<\mu(y,x)$, and arbitrarily if $\mu(x,y)=\mu(y,x)$.
Thus, $T$ becomes an oriented tree with $|V(T)|$ nodes and $|V(T)|-1$ edges, implying that there exists a node $x$ of $T$ that has outdegree $0$.

We first check that $x$ is not a leaf of $T$. Suppose otherwise and let $y$ be the unique neighbor of~$x$.
By~\ref{p:edge} we have that $\mu(y,x)+\mu(x,y)\geq \frac{9}{10}W$. 
Since $xy$ was oriented toward $x$, we infer that $\mu(y,x)\geq \mu(x,y)$, implying $\mu(y,x)\geq \frac{9}{20}W$.
This is a contradiction with~\ref{p:leaf}.

Therefore $x$ is an internal node of $T$. Let $y_1,y_2,y_3$ be the neighbors of $x$.
Without loss of generality suppose $\mu(x,y_1)$ is the largest among $\mu(x,y_1),\mu(x,y_2),\mu(x,y_3)$.
Using~\ref{p:edge} and~\ref{p:node} we infer that
$$\mu(x,y_1)+\mu(x,y_2)+\mu(x,y_3)\geq \frac{9}{10}W-\mu(y_1,x)+\mu(x,y_2)+\mu(x,y_3)\geq \frac{4}{5}W,$$
which implies that $\mu(x,y_1)\geq \frac{4}{15}W>\frac{1}{4}W$. As $xy_1$ was oriented toward $x$, we also have $\mu(x,y_1)\leq \frac{1}{2}W$.
By~\ref{p:edge} we infer that 
$$\mu(y_1,x)\geq \frac{9}{10}W-\mu(x,y_1)\geq \frac{4}{10}W>\frac{1}{4}W\quad\textrm{and}\quad\mu(y_1,x)\leq W-\mu(x,y_1)\leq \frac{11}{15}W<\frac{3}{4}W.$$
Thus, the edge $xy_1$ satisfies the requirements from the statement.
\end{proof}

With Lemma~\ref{lem:balanced-edge} in place, we now give the proof of Lemma~\ref{lem:sc-to-balanced}.

\begin{proof}[Proof of Lemma~\ref{lem:sc-to-balanced}]
We shall say that a closed non-self-crossing curve $\gamma$ on $\Sigma$ is {\em{$\Rr$-light}} if the total weight of objects in $\Rr$ that intersect $\gamma$ is at most $\frac{W}{10}$.

Suppose first that the sphere-cut decomposition $(T,\eta,\delta)$ contains at least one noose $\gamma$ which is not $\Rr$-light.
Then the total weight of objects of $\Rr$ intersected by $\gamma$ is more than $\frac{W}{10}$, 
hence both $\muR(\enc(\gamma))\leq \frac{9}{10}W$ and $\muR(\exc(\gamma))\leq \frac{9}{10}W$ hold, and $\gamma$ satisfies the required properties.
Hence, we may assume that all the nooses in $(T,\eta,\delta)$ are $\Rr$-light.

We now verify that for $\Rr$-light curves as separators, the requirements of Lemma~\ref{lem:balanced-edge} hold.

\begin{claim}\label{lem:light-almost-all}
Suppose $\gamma$ is an $\Rr$-light curve on $\Sigma$. Then 
$$\frac{9}{10}W \leq \muR(\enc(\gamma))+\muR(\exc(\gamma))\leq W.$$
\end{claim}
\begin{proof}
The right inequality is obvious. For the left one, all objects in $\Rr$ that do not intersect $\gamma$ are entirely contained in either
in $\enc(\gamma)$ or in $\exc(\gamma)$, thus they contribute their weight to the sum $\muR(\enc(\gamma))+\muR(\exc(\gamma))$, while the total weight of objects intersecting $\gamma$ is at most $\frac{W}{10}$
due to $\Rr$-lightness of $\gamma$.
\cqed\end{proof}

\begin{claim}\label{lem:light-sum}
Suppose $\gamma,\gamma_1,\gamma_2$ are $\Rr$-light curves such that $\enc(\gamma)$ is equal to the disjoint union of $\enc(\gamma_1)$, $\enc(\gamma_2)$ and $(\gamma_1\cap \gamma_2)\setminus \gamma$.
Then
$$\muR(\enc(\gamma_1))+\muR(\enc(\gamma_2))\leq \muR(\enc(\gamma))\leq \muR(\enc(\gamma_1))+\muR(\enc(\gamma_2))+\frac{W}{10}.$$
\end{claim}
\begin{proof}
The left inequality is obvious, because every subset from $\Rr$ entirely contained either in $\enc(\gamma_1)$ or in $\enc(\gamma_2)$ is also entirely contained in $\enc(\gamma)$, while $\enc(\gamma_1)$ and
$\enc(\gamma_2)$ are disjoint. For the right inequality, observe that every object from $\Rr$ that is entirely contained in $\enc(\gamma)$ but not in $\enc(\gamma_1)$ or $\enc(\gamma_2)$ has to intersect
$(\gamma_1\cap \gamma_2)\setminus \gamma$. Since $\gamma_1$ is $\Rr$-light, the total weight of such subsets is at most $\frac{W}{20}$.
\cqed\end{proof}

We proceed with the proof of Lemma~\ref{lem:sc-to-balanced}.
For any ordered pair $(x,y)$ of adjacent nodes in $T$, let $\mu(x,y)=\muR(\enc(\delta(x,y))$.
Since all nooses in $(T,\eta,\delta)$ are $\Rr$-light and $(T,\eta,\delta)$ is faithful, 
Claims~\ref{lem:light-almost-all} and~\ref{lem:light-sum} verify that conditions~\ref{p:edge} and~\ref{p:node} of Lemma~\ref{lem:balanced-edge} hold.
Condition~\ref{p:leaf}, on the other hand, is satisfied by the assumptions of the lemma.
Therefore, we may apply Lemma~\ref{lem:balanced-edge}, yielding a pair $(x,y)$ such that $\delta(x,y)$ satisfies the required assertions.
\end{proof}


\subsection{From nooses to separators: proof of Lemma~\ref{lem:balanced-sep}}\label{sec:nooses-to-seps}

For the proof of Lemma~\ref{lem:balanced-sep} we need to formalize the connection between nooses in the Voronoi diagram and Voronoi separators in the graph.
This connection was largely explored in~\cite{MarxP15}. We now recall and adjust the basic observations; see Figure~\ref{fig:vorsep} for a visualization.

\begin{figure}
                \centering
                \def\svgwidth{0.8\columnwidth}
                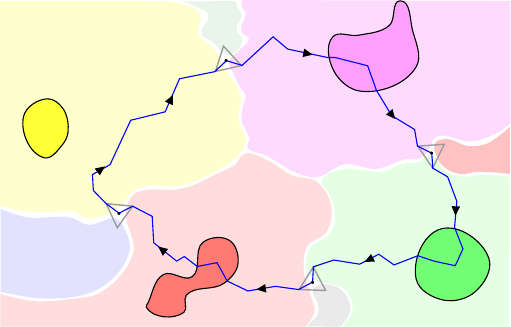
\caption{A Voronoi separator $S$ of length $4$ and its perimeter $\pi(S)$, depicted in blue. The Voronoi regions are depicted in light colors, objects in respective solid colors. 
As in Figure~\ref{fig:diamonds}, the perimeter $\pi(S)$ may not cross some object $p$ traversed by $S$.
The figure is taken almost verbatim from~\cite{MarxP15},
with slight modifications.}\label{fig:vorsep}
\end{figure}

\begin{figure}
                \centering
                \def\svgwidth{0.8\columnwidth}
                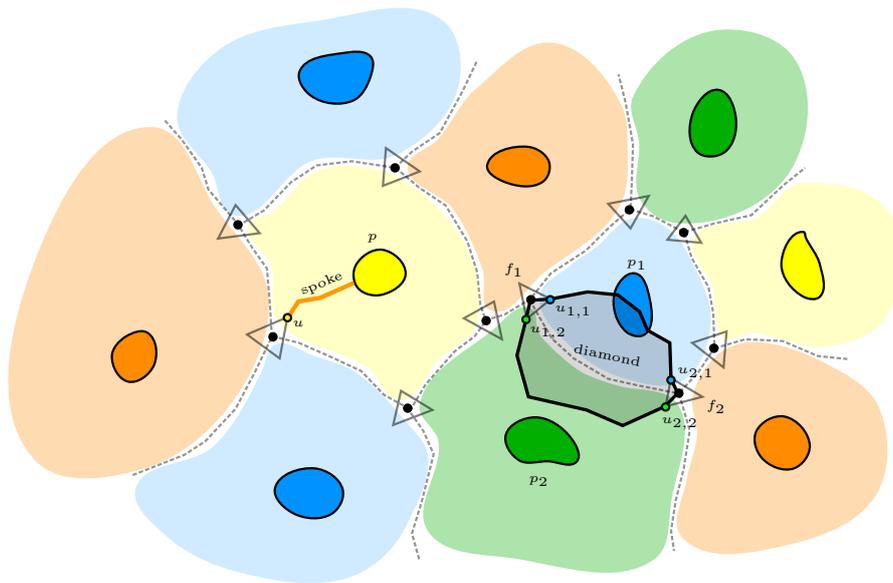
\caption{Part of the Voronoi diagram with a spoke and a diamond highlighted. Depicted triangles are the branching points of the Voronoi diagram. 
Gray dashed curves are the edges of the Voronoi diagram. The spoke of $u$, which is the shortest path leading from $u$ to the object $p$ to whose region $u$ belongs, is depicted in orange.
The diamond induced by the edge $f_1f_2$ is depicted in black (its interior is grayed). Note that the diamond does not need to cross objects $p_1$ and $p_2$.
For instance, as in the figure, it may happen that
the path within $\widehat{T}(p_2)$ connecting $u_{2,1}$ and $u_{2,2}$ does not intersect~$p_2$.}\label{fig:diamonds}
\end{figure}

Suppose $G$ is an edge-weighted connected graph embedded in a sphere $\Sigma$, $\Ss\subseteq \Ff$ are independent families of objects in $G$ with $|\Ss|\geq 4$, and $H=H_{\Ss}$ is the Voronoi diagram of $\Ss$ in $G$. 
Recall that $H$ is a connected $3$-regular multigraph embedded in $G$. Suppose $\gamma$ is a noose in $H$. Then $\gamma$ naturally induces a Voronoi separator $S(\gamma)$ defined as the sequence: 
\begin{equation}\label{eq:sep}
\langle p_1,u_1,f_1,v_1,p_2,u_2,f_2,v_2,\ldots,p_r,u_r,f_r,v_r\rangle
\end{equation}
such that 
\begin{itemize}
\item $f_1,\ldots,f_r$ are consecutive branching points of $H$ visited by $\gamma$;
\item between $f_{i-1}$ and $f_i$, $\gamma$ travels through the face of $H$ corresponding to the object $p_i$;
\item $u_i$ is the vertex of $f_i$ corresponding to the direction of $\gamma$ entering $f_i$; and
\item $v_i$ is the vertex of $f_i$ corresponding to the direction of $\gamma$ leaving $f_i$.
\end{itemize}
Here, by correspondence between vertices of $f_i$ and directions entering/leaving $f_i$ we mean the following. Recall that $f_i$ is a triangular face of $G$ and at the same time it is a degree-$3$ vertex of $H$.
Edges incident to $f_i$ in $H$ are in the dual correspondence to the edges of $f_i$ in $G$, thus the vertices of $f_i$ may be naturally associated with the three incidences between $f_i$ and faces of $H$ around it:
vertex of $f_i$ in $G$ incident to edges $e_1,e_2$ of $f_i$ in $G$ corresponds to the incidence with the face lying between the duals of $e_1,e_2$ in $H$.
A noose in $H$ may enter/leave $f_i$ in three directions corresponding to these three incidences, hence we have a correspondence between the vertices of $f_i$ and the directions.

For a Voronoi separator $S$ as in \eqref{eq:sep}, define its {\em{perimeter}} $\pi(S)$ to be the noose constructed by concatenating
the following $2r$ curves in order, two for each $i=1,2,\ldots,r$:
\begin{itemize}
\item the unique path in $\widehat{T}(p)$ from $v_{i-1}$ to $u_i$, and
\item a curve within $f_i$ from $u_i$ to $v_i$, obtained by concatenating the segment from $u_i$ to the center of $f_i$ and the segment from the center of $f_i$ to $v_i$.
\end{itemize}

Given a noose $\gamma$ with respect to $H$, the {\em{canonical version}} of $\gamma$ is the noose $\pi(S(\gamma))$.
Note that $\gamma$ and its canonical version are equivalent as nooses with respect to $H$.
A noose that is equal to its own canonical version shall be called {\em{canonical}}.
It can be easily seen that if $\gamma$ is a canonical noose with respect to $H$, then $\enc(\gamma)$ is the union of interiors of diamonds induced by edges of $H$ enclosed by $\gamma$; symmetrically for $\exc(\gamma)$.
Then the following is immediate.

\begin{lemma}\label{prop:canon}
Every sphere-cut decomposition of $H$ whose all nooses are canonical is faithful.
\end{lemma}

We are ready to give a proof of Lemma~\ref{lem:balanced-sep}.

\begin{proof}[Proof of Lemma~\ref{lem:balanced-sep}]
Apply Lemma~\ref{lem:sampling} to $\Ff$ and 
$\ell=s^2=10^6\cdot \left(\frac{1}{\eps}\ln \frac{1}{\eps}\right)^2$.
Note that thus $\ell\geq 10$, as required by Lemma~\ref{lem:sampling}.
This yields a subfamily $\Ss\subseteq \Ff$ and the corresponding Voronoi diagram $H=H_{\Ss}$ satisfying the following properties:
\begin{enumerate}[label=(N\arabic*),ref=(N\arabic*)]
\item $4\leq |\Ss|\leq 2\ell$;
\item\label{p:spoke} every spoke in $H$ has weight with respect to $\Ff$ upper bounded by $10\ln \ell \cdot \frac{W}{\ell}$; and
\item\label{p:diamond} every diamond in $H$ has weight with respect to $\Ff$ upper bounded by $10\ln \ell \cdot \frac{W}{\ell}$.
\end{enumerate}
Recall that $H$ is a $3$-regular connected multigraph embedded in $\Sigma$.
Let us continue the proof under the assumption that $H$ is bridgeless.
The general case when $H$ may have bridges requires attention to a few more technical details but is conceptually no different; 
we present it later. 

Since $H$ is bridgeless, by Theorem~\ref{thm:sc-decomp} we may find a sphere-cut decomposition $(T,\eta,\delta)$ of $H$ of width at most $\sqrt{4.5 |V(H)|}\leq \sqrt{9\ell}=3\sqrt{\ell}$.
We may further assume that each noose used by $(T,\eta,\delta)$ is canonical, hence $(T,\eta,\delta)$ is faithful by Lemma~\ref{prop:canon}.
The following claims now follow easily from the properties of $\Ss$ provided by Lemma~\ref{lem:sampling}.

\begin{claim}\label{cl:noose-weight}
For every pair $(x,y)$ of adjacent nodes of $T$, the total weight of objects of $\Ff$ banned by the Voronoi separator $S(\delta(x,y))$ is at most $\eps W$. 
\end{claim}
\begin{proof}
Since the width of $(T,\eta,\delta)$ is at most $3\sqrt{\ell}$, the noose $\delta(x,y)$ traverses at most $3\sqrt{\ell}=3s$ faces of $H$.
Each traversal, say of the face corresponding to an object $p\in \Ss$, consists of the union of two spokes and a path within $p$.
Observe that the objects banned by $S(\delta(x,y))$ are exactly the objects that are in conflict with any of these at most $6\sqrt{\ell}$ spokes.
By~\ref{p:spoke}, each of these spokes has weight with respect to $\Ff$ upper bounded by $10\ln \ell \cdot \frac{W}{\ell}$, which means that the total weight of objects banned by $S(\delta(x,y))$ is at most
$$6\sqrt{\ell}\cdot 10\ln \ell \cdot \frac{W}{\ell}=60\ln \ell \cdot \frac{W}{\sqrt{\ell}}=120\ln s\cdot \frac{W}{s}\leq \eps W,$$
where the last inequality holds due to $\ln s\leq 7\cdot \ln \frac{1}{\eps}$ being satisfied for $\eps<\frac{1}{10}$.\cqed\end{proof}

\begin{claim}\label{cl:leaves}
If $(x,y)$ is a pair of adjacent nodes in $T$ such that $x$ is a leaf, then the total weight of objects of $\Ff$ contained in $\enc(\delta(y,x))$ is at most $\eps W$.
\end{claim}
\begin{proof}
Let $e=\eta^{-1}(x)$. Then $\delta(y,x)$ is the diamond induced by $e$ and $\enc(\delta(y,x))$ is its interior. By~\ref{p:diamond}, every diamond in $H$ has weight with respect to $\Ff$ upper bounded by
$$10\ln \ell \cdot \frac{W}{\ell}\leq \eps W,$$
similarly as in Claim~\ref{cl:noose-weight}. The claim follows.
\cqed\end{proof}

The next claim is the key insight from~\cite{MarxP15}: Voronoi separators induced by nooses in $(T,\eta,\delta)$ break the intersection graph $\IntGraph(\Ff)$ in the expected way.

\begin{claim}\label{cl:separability}
Suppose $(x,y)$ is a pair of adjacent nodes of $T$ and $p,p'\in \Ff$ are two objects not banned by $S(\delta(x,y))$ such that $p$ and $p'$ are contained in different regions of $\Sigma-\delta(x,y)$.
Then $p$ and $p'$ are non-adjacent in $\IntGraph(\Ff)$.
\end{claim}
\begin{proof}
Since $p$ and $p'$ are contained in different regions of $\Sigma-\delta(x,y)$, they are in particular disjoint, and hence they are non-adjacent in the intersection graph $\IntGraph(\Ff)$.
\cqed\end{proof}

For an open disk $\Delta$ on $\Sigma$, define the measure $\mu_{\Ff}(\Delta)$ as before: $\mu_{\Ff}(\Delta)$ is the sum of weights of objects $p\in \Ff$ that are entirely contained in $\Delta$.
Since $\eps\leq \frac{9}{20}$, Claim~\ref{cl:leaves} verifies that the prerequisites of Lemma~\ref{lem:sc-to-balanced} are satisfied.
Thus, we may apply Lemma~\ref{lem:sc-to-balanced} to $(T,\eta,\delta)$, yielding a pair of adjacent nodes $(x,y)$ of $T$ such that within each of the disks $\enc(\delta(x,y))$ and $\enc(\delta(y,x))$ the total
weight of objects from $\Ff$ is at most $\frac{9}{10}W$.

We claim that the Voronoi separator $S$ satisfies all the required conditions.
For condition~\ref{p:important}, we have that $\Dd(S)\subseteq \Ss\subseteq \Ff$ and all faces traversed by $S$ are branching points of the Voronoi diagram of $\Ss\subseteq \Dd$, hence they are $\Dd$-important.
For condition~\ref{p:length}, the length of $S(\delta(x,y))$ is equal to the length of $\delta(x,y)$, which is at most $3s$.
Condition~\ref{p:sacrifice} follows directly from Claim~\ref{cl:noose-weight}.
Finally, for condition~\ref{p:balanced}, Claim~\ref{cl:separability} together with the construction of $S$
imply that every connected component of $\IntGraph(\Dd)-\Ban(S)$ contains objects from $\Ff$ of total weight at most at most $\frac{9}{10}W$.
Hence $S$ has all the prescribed properties.

%

\medskip

We are left with discussing the case when the Voronoi diagram $H$ has bridges.
Note that diamonds in $H$ are well-defined also for bridges; hence, as in Claim~\ref{cl:leaves},
every diamond in $H$ has weight at most $10\ln \ell\cdot \frac{W}{\ell}\leq \eps W$.

Consider any bridge $b$ in $H$, say with endpoints $f_1$ and $f_2$; recall that these are branching points of $H$.
For $t=1,2$ let us construct the canonical noose (w.r.t. $H$) of length $1$ traveling through $f_t$ and the unique face of $H$ incident to the bridge $b$.
Call this noose $\beta(f_t,b)$ and orient it so that $\enc(\beta(f_t,b))$ contains the bridge $b$; nooses constructed in this manner shall be called {\em{bridge nooses}}.

Same calculations as in Claim~\ref{cl:noose-weight} show that for any bridge noose $\beta$, the total weight of objects of $\Ff$ banned by the Voronoi separator $S(\beta)$ is at most $\eps W$;
this is because the length of a bridge noose is $1$, which is never larger than $3s$.
This means that if for any bridge noose $\beta$ we have that for both $\enc(\beta)$ and $\exc(\beta)$, the total weight of objects from $\Ff$ contained in each of these disks is at
most $\frac{9}{10}W$, then $S(\beta)$ satisfies all the required conditions.
Hence, from now on we assume that for every bridge noose $\beta$, either $\mu_{\Ff}(\enc(\beta))$ or $\mu_{\Ff}(\exc(\beta))$ exceeds $\frac{9}{10}W$.

Let us inspect any bridge $b$ in $H$, say with endpoints $f_1$ and $f_2$. Denote $\beta_1=\beta(f_1,b)$ and $\beta_2=\beta(f_2,b)$.
As $\exc(\beta_1)$ and $\exc(\beta_2)$ are disjoint, it cannot happen that $\mu_\Ff(\exc(\beta_1))>\frac{9}{10}W$ and $\mu_\Ff(\exc(\beta_2))>\frac{9}{10}W$.
Also, if it happened that $\mu_\Ff(\enc(\beta_1))>\frac{9}{10}W$ and $\mu_\Ff(\enc(\beta_2))>\frac{9}{10}W$, 
then the objects from $\Ff$ contained in the interior of the diamond induced by $b$ in $H$ (which is equal to $\enc(\beta_1)\cap \enc(\beta_2)$) would have weight more than $\frac{8}{10}W$, 
a contradiction with the fact that every diamond in $H$ has weight at most $\eps W$.
Therefore, either $\mu_\Ff(\enc(\beta_1))>\frac{9}{10}W$ and $\mu_\Ff(\exc(\beta_2))>\frac{9}{10}W$, or $\mu_\Ff(\exc(\beta_1))>\frac{9}{10}W$ and $\mu_\Ff(\enc(\beta_2))>\frac{9}{10}W$
Orient $b$ toward $f_2$ in the former case and toward $f_1$ in the latter case.

Thus, every bridge of $H$ becomes oriented. Since $H$ has one fewer bridge than bridgeless component, there exists a bridgeless component $B$ of $H$ such that every bridge incident to $B$ is oriented toward its
endpoint residing in $B$.

We first resolve the corner case when $B$ consists of one vertex, say $u$. Since $H$ is $3$-regular, $u$ has three different neighbors $v_1,v_2,v_3$ such that $uv_1,uv_2,uv_3$ are the three bridges incident to $u$.
These bridges are oriented toward $u$, which means that each of the disks $$\exc(\beta(u,uv_1)),\exc(\beta(u,uv_2)),\exc(\beta(u,uv_3))$$ contains objects from $\Ff$ of total weight 
more than $\frac{9}{10}W$. However,
each object from $\Ff$ can be contained in at most $2$ of these disks, and a simple counting argument yields a contradiction.

Hence, from now on suppose $B$ is a bridgeless component of $H$ consisting of more than one vertex.
Since $B$ is bridgeless and nontrivial, every vertex of $u$ of $B$ is incident to at most one bridge of $H$.
Thus, it is easy to see that there is an injective mapping $\phi$ from bridges incident to $B$ to edges of $B$ so that each bridge $b$ shares an endpoint with $\phi(b)$.

For a bridge $b$ incident to $B$, let $H_b$ be the component of $H-b$ that is disjoint with $B$.
Now, for every noose $\gamma$ with respect to $B$ we may define a noose $\hat{\gamma}$ with respect to $H$ so that $\gamma$ and $\hat{\gamma}$ are equivalent as nooses w.r.t. $B$, 
and for every bridge $b$ incident to $B$, either all of $H_b$, $b$, and $\phi(b)$ are contained in $\enc(\gamma)$, or all of them are contained in $\exc(\gamma)$.
See Figure~\ref{fig:shift} for an example.
We may further require that $\hat{\gamma}$ is canonical (as a noose with respect to $H$).

\begin{figure}[h!]
                \centering
                \def\svgwidth{0.98\columnwidth}
                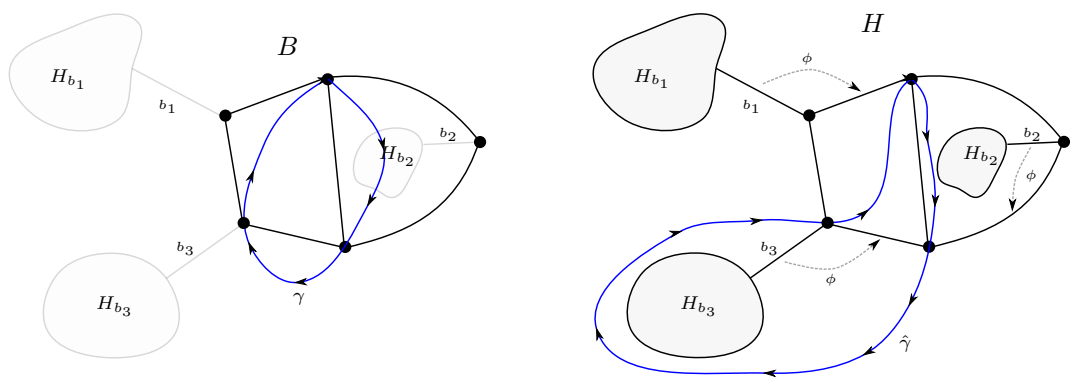
\caption{Construction of $\hat{\gamma}$ (a noose w.r.t. $H$) out of $\gamma$ (a noose w.r.t $B$). The left panel depicts the bridgeless component $B$ alone together with the noose $\gamma$.
The right panel depicts the whole graph $H$, with components $H_{b_1},H_{b_2},H_{b_3}$ and bridges $b_1,b_2,b_3$ reintroduced, and the adjusted noose~$\hat{\gamma}$.
The mapping $\phi$ is depicted with gray, dashed arrows.}\label{fig:shift}
\end{figure}

By Theorem~\ref{thm:sc-decomp}, let $(T,\eta,\delta)$ be a sphere-cut decomposition of $B$ of width at most $\sqrt{4.5|V(B)|}\leq 3\sqrt{\ell}$.
Replace each noose $\gamma$ appearing in $(T,\eta,\delta)$ by $\hat{\gamma}$; it is easy to see that thus $(T,\eta,\delta)$ is still a sphere-cut decomposition of $B$ and becomes faithful.
We would like now to apply the whole reasoning to $(T,\eta,\delta)$ and $B$.
Claim~\ref{cl:noose-weight} works exactly as before, because the only condition we needed is an upper bound on the weights of spokes in $H$.

However, Claim~\ref{cl:leaves} may fail,
because disks of the form $\enc(\delta(y,x))$, where $x$ is a leaf of $T$ and $y$ is the node of $T$ adjacent to it,
are no longer diamonds in $H$.
More precisely, if $e=\eta^{-1}(x)$ is the edge of $H$ associated with $x$ and $\Delta(e)$ is the interior of the diamond of $H$ induced by it, then the following two cases may happen:
\begin{itemize}
\item if $e$ is not in the image of $\phi$, then $\enc(\delta(y,x))=\Delta(e)$;
\item if $e=\phi(b)$ for some bridge $b$ incident to $B$, say with $u$ being the endpoint of $b$ in $B$, then $\enc(\delta(y,x))$ is the union of $\Delta(e)$,
the disk $\enc(\beta(u,b))$, and the intersection of their boundaries.
\end{itemize}
In the first case, we already know that $\mu_\Ff(\enc(\delta(y,x)))\leq \eps W<\frac{9}{20}W$.
In the second case, observe that every object of $\Ff$ contained in $\enc(\delta(y,x))$ is either contained $\Delta(e)$ or is not contained in $\exc(\beta(u,b))$.
The total weight of objects satisfying the first condition is at most $\eps W$, whereas the total weight of objects satisfying the second condition is less than $\frac{W}{10}$, because 
$\mu_\Ff(\exc(\beta(u,b)))>\frac{9}{10}W$. Hence, $$\mu_{\Ff}(\enc(\delta(y,x)))\leq\eps W + \frac{W}{10}<\frac{9}{20}W.$$
Summing up, in both cases we have verified that $\mu_{\Ff}(\enc(\delta(y,x)))<\frac{9}{20}W$, so we may apply Lemma~\ref{lem:sc-to-balanced} and proceed as in the case without bridges.
\end{proof}

\section{A QPTAS for {\sc{Maximum Weight Independent Set of Objects}}}\label{sec:mwis}
In this section we use Lemma~\ref{lem:breaking} to design a QPTAS for {\sc{MWISO}}, that is, we prove Theorem~\ref{thm:main-mwis}.
Recall that the setting is as follows. The input is $(G,\mathcal{D})$, where $G$ is a graph embedded in a sphere $\Sigma$ together with a family of objects $\mathcal{D}$, 
each being a connected subgraph of $G$ with a prescribed positive weight. Moreover, we are given an accuracy parameter $\epsilon>0$ and we may assume w.l.o.g. that $\epsilon<\frac{1}{10}$.
The goal is to find an independent subfamily $\mathcal{F}\subseteq \mathcal{D}$ with the largest possible weight;
more precisely, the algorithm shall compute a solution of weight at least $(1-\epsilon)$ times the optimum.
Let $n=|V(G)|$ and $N=|\mathcal{D}|$; w.l.o.g. we assume $N\geq 2$.

First, we assume that all objects in $\mathcal{D}$ have weights between $1$ and $M=2\epsilon^{-1}N$.
For the given instance $(G,\Dd)$ let us define
$$\spn(\Dd)=\frac{\max_{p\in \Dd}\weight(p)}{\min_{p\in \Dd}\weight(p)}.$$
The following standard claim shows that we may assume that the span in the input instance is bounded polynomially in $N$ and $1/\eps$; more precisely, by $M=2\eps^{-1}N$.

\begin{claim}\label{cl:span-reduction}
Suppose there is a QPTAS as described in Theorem~\ref{thm:main-mwis}, but working on instances $(G,\Dd)$ that additionally satisfy $\spn(\Dd)<2\eps^{-1}|\Dd|$.
Then Theorem~\ref{thm:main-mwis} holds.
\end{claim}
\begin{proof}
Let $M=2\eps^{-1}|\Dd|$.
For every object $p\in \Dd$ construct a subfamily $\Dd_p\subseteq \Dd$ consisting of all objects whose weight is not larger than $\weight(p)$ but larger than $\weight(p)/M$.
Clearly, $\spn(\Dd_p)<M$.
We claim that there exists $p\in \Dd$ such that the optimum solution for $(G,\Dd_p)$ has weight at most $(1-\eps/2)$ smaller than the optimum solution for $(G,\Dd)$.
If this is the case, then it suffices to run the assumed QPTAS on each of the instances $(G,\Dd_p)$, for $p\in \Dd$, using accuracy parameter $\eps/2$, and output the heaviest of the solutions found.

Let $\Fopt\subseteq \Dd$ be the optimum solution for $(G,\Dd)$. Let $p$ be the heaviest object in $\Fopt$.
Consider $\Ff=\{q\in \Fopt\colon \weight(q)>\weight(p)/M\}$ and observe that $\Ff\subseteq \Dd_p$, thus $\Ff$ is a solution for $(G,\Dd_p)$.
On the other hand, we have
$$\weight(\Fopt\setminus \Ff)\leq \weight(p)/M\cdot |\Fopt\setminus \Ff|\leq \weight(p)/M\cdot |\Dd|=\eps/2\cdot \weight(p)\leq \eps/2\cdot \weight(\Fopt),$$
where the last inequality follows from $p\in \Fopt$. Hence $\weight(\Ff)\geq (1-\eps/2)\weight(\Fopt)$ and we are done.
\cqed\end{proof}

By Claim~\ref{cl:span-reduction} in order to prove Theorem~\ref{thm:main-mwis} it suffices to work under the assumption that $\spn(\Dd)<M$.
Hence, by rescaling the weights we may assume that all the weights of objects in $\Dd$ are at least $1$ and smaller than $M$.
%
%

Before we proceed to the algorithm, we fix the following parameters:
\begin{equation*}
d_{\max}=10\ln(MN),\qquad \qquad \hat{\epsilon} = \frac{\epsilon}{d_{\max}},\qquad \qquad s=10^3\cdot\frac{1}{\hat{\epsilon}} \ln \frac{1}{\hat{\epsilon}}.
\end{equation*}
Parameter $d_{\max}$ is the maximum recursion depth of the algorithm; note that $d_{\max}=\mathcal{O}(\log (N/\epsilon))$.
Next, $\hat{\epsilon} = \mathcal{O}(\epsilon/\log(N))$ is the refined accuracy parameter which will be used throughout the recursion instead of $\epsilon$.
Similarly as in~\cite{AW2013,adamaszek2014qptas,Har-Peled2014}, intuitively we lose a factor of $1-\hat{\epsilon}$ in each recursion level which
yields an overall approximation ratio of $(1-\hat{\epsilon})^{d_{\max}}=(1- \epsilon/\log(N))^{\mathcal{O}(\log(N))}=1-\mathcal{O}(\epsilon)$.
Let us stress that although the algorithm uses recursion and the number of objects changes in subsequent recursive calls, the values of $d_{\max},\hat{\epsilon},s$ are fixed as above 
and their definitions always refer to the initial number of objects.

We now explain the algorithm; it is also summarized using pseudocode as Algorithm~\ref{alg:main}.
Let us fix an optimum solution $\mathcal{F}_{\mathrm{OPT}}$ and denote $W=\mathbf{w}(\mathcal{F}_{\mathrm{OPT}})$. 
We shall analyze $\mathcal{F}_{\mathrm{OPT}}$, which will lead to the formulation of the algorithm as a recursive search for $\mathcal{F}_{\mathrm{OPT}}$.

We would like to use Lemma~\ref{lem:breaking} to guess a Voronoi separator that breaks $\mathcal{F}_{\mathrm{OPT}}$ in a balanced way.
However, we first need make sure that every object in question constitutes only a small fraction of $W$.
This is done by a standard method of guessing exactly ``heavy'' objects in the solution, whose number is small, and proceeding only with the ``light'' ones.

More precisely, call an object $p\in \mathcal{F}_{\mathrm{OPT}}$ {\em{heavy}} if $\mathbf{w}(p)>s^{-2}W$.
Observe that the number of heavy objects in $\mathcal{F}_{\mathrm{OPT}}$ is at most $s^2$, hence there are at most $N^{s^2}$ possible ways to select those heavy objects from $\mathcal{D}$.
The algorithm branches into all possible such ways, in each branch fixing a different candidate for the set of heavy objects.
Hence, by increasing the number of subproblems by a multiplicative factor $N^{s^2}$ we may assume that the algorithm fixes the set $\mathcal{F}_{\mathrm{hv}}$ 
consisting of all heavy objects in $\mathcal{F}_{\mathrm{OPT}}$. Let $\mathcal{D}'$ be obtained from $\mathcal{D}$ by removing $\mathcal{F}_{\mathrm{hv}}$ and all objects intersecting any object from $\mathcal{F}_{\mathrm{hv}}$, and let $\mathcal{F}_{\mathrm{OPT}}'=\mathcal{F}_{\mathrm{OPT}}-\mathcal{F}_{\mathrm{hv}}$.
Note that $\mathbf{w}(\mathcal{F}_{\mathrm{OPT}}')\leq W$ and $\mathbf{w}(p)\leq s^{-2}W$ for all $p\in \mathcal{F}_{\mathrm{OPT}}'$.

We may now apply Lemma~\ref{lem:breaking} to $\mathcal{F}_{\mathrm{OPT}}'\subseteq \mathcal{D}'$ with $W$ being the upper bound on its weight. 
Thus, in time $N^{\mathcal{O}(s)}\cdot n^{\mathcal{O}(1)}$ we may compute a family $\mathbb{X}$ consisting of subsets of $\mathcal{D}'$ with $|\mathbb{X}|\leq 6^{3s}N^{15s}$ and satisfying the following property:
there exists $\mathcal{X}\in \mathbb{X}$ such that $\mathbf{w}(\mathcal{F}_{\mathrm{OPT}}'\setminus \mathcal{X})\leq \epsilon W$ and 
within each connected component of the graph $\mathrm{IntGraph}(\mathcal{D}')-\mathcal{X}$ the total weight of objects from $\mathcal{F}_{\mathrm{OPT}}'$ does not exceed $\frac{9}{10}W$.
By branching into all the members of $\mathbb{X}$, via increasing the number of subproblems by a multiplicative factor $|\mathbb{X}|$ we may henceforth assume that the algorithm fixes $\mathcal{X}$ with properties as above.

For a fixed choice of $\mathcal{F}_{\mathrm{hv}}$ and $\mathcal{X}$ as above, let us inspect the connected components of $\mathrm{IntGraph}(\mathcal{D}')-\mathcal{X}$; let their vertex sets be $\mathcal{D}_1,\ldots,\mathcal{D}_k$. 
We apply the algorithm recursively to the instances $(G,\mathcal{D}_i)$ for $i=1,\ldots,k$, yielding independent families $\mathcal{F}_1,\ldots,\mathcal{F}_k$ with $\mathcal{F}_i\subseteq \mathcal{D}_i$. 
We record the family
$\mathcal{F}=\mathcal{F}_{\mathrm{hv}}\cup \bigcup_{i=1}^k \mathcal{F}_i$
as a candidate for the solution; it is straightforward to see from the construction that this family is independent. 
Finally, as the final solution we output the heaviest among the recorded candidates; that is,
the heaviest solution found for all choices of $\mathcal{F}_{\mathrm{hv}}$ and $\mathcal{X}$.
We remark that if for some choice of $\mathcal{F}_{\mathrm{hv}}$ it turned out that $\mathcal{D}'=\emptyset$, i.e., every object intersects some objects from $\mathcal{F}_{\mathrm{hv}}$, then $\mathbb{X}$ contains only one choice of $\mathcal{X}$ being $\emptyset$, 
hence we include $\mathcal{F}=\mathcal{F}_{\mathrm{hv}}$ among the candidates without invoking any recursive~calls.

The base case of the recursion is provided by trimming it at level $d_{\max}$. 
More precisely, all subcalls at depth larger than $d_{\max}$ return empty solutions.
This concludes the description of the algorithm; as mentioned, it is summarized using pseudocode as Algorithm~\ref{alg:main}.

\begin{algorithm}
  \KwIn{An instance $(G,\mathcal{D})$, recursion depth $d$} 
  \KwOut{An independent family $\mathcal{F}\subseteq \mathcal{D}$}  \Indp \BlankLine
  \If{$d>d_{\max}$}{\KwRet{$\emptyset$}}
  $\mathcal{F}\leftarrow \emptyset$\\
  \ForAll{\normalfont{$\mathcal{F}_{\mathrm{hv}}$: independent subfamily of $\mathcal{D}$ with $|\mathcal{F}_{\mathrm{hv}}|\leq s^2$}}{
        $\mathcal{D}'\leftarrow \mathcal{D} - (\textrm{all objects that intersect any object of}\ \mathcal{F}_{\mathrm{hv}})$\\
        $\mathbb{X}\leftarrow$ family computed for $\mathcal{D}'$ using Lemma~\ref{lem:breaking}\\
        \ForAll{$\mathcal{X}\in \mathbb{X}$}{
	  $\mathcal{D}_1,\ldots,\mathcal{D}_k\leftarrow$ vertex sets of the connected components of $\mathrm{IntGraph}(\mathcal{D}')-\mathcal{X}$\\
	  \For{$i=1\ \mathbf{to}\  k$}{
	    $\mathcal{F}_i\leftarrow \mathtt{AlgMWISO}(G,\mathcal{D}_i,d+1)$
	  }
	  $\mathcal{F}_{\textrm{cand}}\leftarrow \mathcal{F}_{\mathrm{hv}}\cup \bigcup_{i=1}^k \mathcal{F}_i$\\
	  \If{$\mathbf{w}(\mathcal{F}_{\mathrm{cand}})>\mathbf{w}(\mathcal{F})$}{$\mathcal{F}\leftarrow \mathcal{F}_{\textrm{cand}}$}
        }
   }
   \KwRet{$\mathcal{F}$}  
   
\caption{Algorithm $\mathtt{AlgMWISO}$}
  \label{alg:main}
\end{algorithm}

\subparagraph*{Running time analysis.}
On a high level, the above algorithm runs in time $2^{\mathrm{poly}(1/\epsilon,N)}\cdot n^{\mathcal{O}(1)}$ because at each node of the recursion tree the algorithm uses polynomial time and calls itself on $N^{\mathcal{O}(s^2)}$ subinstances, where $s=\mathrm{poly}(1/\epsilon,\log N)$.
Together with the bound of $d_{\max}=\mathcal{O}(\log (N/\epsilon))$ on the recursion depth this yields the promised running time. 

Formally, we argue as follows. Throughout this analysis $N$ denotes the initial number of objects.
It is straightforward to see that the computation at each recursive call, excluding time spent on executing the invoked subcalls, takes time $N^{\Oh(s)}\cdot n^{\Oh(1)}$. 
Therefore, to bound the running time it suffices to bound the total number of nodes in the recursion tree.
As the recursion is trimmed at depth $\dmax$, its height is at most $\dmax$.
Further, in each call we investigate:
\begin{itemize}
\item at most $N^{s^2}$ ways to select heavy objects $\Fheavy$ from $\Dd$;
\item at most $6^{3s}N^{15s}$ ways to select the family $\Xx$; and
\item at most $N$ connected components of $\IntGraph(\Dd')-\Xx$.
\end{itemize}
This yields at most $6^{3s}\cdot N^{s^2+18s+1}\leq N^{2s^2}$ instances on which a recursive subcall is invoked.
This means that each node of the recursion tree has at most $N^{2s^2}$ children, which together with the height of at most $\dmax=10\ln (NM)$ yields an upper bound of $N^{20s^2\ln (NM)}$ on the number of nodes 
in the recursion tree.
Since the computation at each node takes time $N^{\Oh(s)}\cdot n^{\Oh(1)}$, we arrive at the total running time of
$$N^{\Oh(s^2\log (NM))}\cdot n^{\Oh(1)}=2^{\poly(1/\eps,N)}\cdot n^{\Oh(1)},$$
as promised.

\subparagraph*{Approximation factor analysis.}
As mentioned before, 
the claimed approximation ratio follows as intuitively in each of the $d_{\max}$ recursion levels we lose a factor of $1-\hat{\epsilon}$, accumulating to 
$(1-\hat{\epsilon})^{d_{\max}}=1-\mathcal{O}(\epsilon)$ overall.
Formally, we argue as follows. It is clear that the family output by the algorithm is indeed independent.
In order to prove that its weight is at least $(1-\eps)W$, it suffices to show the following claim and apply it for $d=0$; i.e., for the initial instance $(G,\Dd)$.
Recall here that $\Fopt$ is a fixed optimum solution for the whole instance.

\begin{claim}
Let $(G,\td{\Dd})$ be an instance on which the algorithm is called in the recursion tree, say at depth $d$.
Suppose further that $\weight(\Fopt\cap \td{\Dd})\leq (9/10)^d\cdot W$.
Then the call of the algorithm to $(G,\td{\Dd})$ returns an independent subfamily $\td{\Ff}\subseteq \td{\Dd}$ with $\weight(\td{\Ff})\geq (1-d'\epsloc)\weight(\Fopt\cap \td{\Dd})$, where $d'=\dmax-d$.
\end{claim}
\begin{proof}
Denote $\tdFopt=\Fopt\cap \td{\Dd}$.
We prove the claim by induction on the depth $d$, starting with the case $d\geq \dmax$ and then decreasing $d$.
For the base of the induction, for $d\geq \dmax$ we have
$$(9/10)^d\leq (1-1/10)^{10\ln (MN)}<e^{-\ln (MN)}=\frac{1}{MN}.$$
On the other hand, $\Fopt$ contains at most $N$ objects of weight less than $M$ each, hence $W<MN$.
Combining these two observations we infer that $(9/10)^{d}\cdot W<1$.
Since each object of $\Dd$ has weight at least $1$, we conclude that $\tdFopt=\emptyset$, so the claim holds vacuously.

Suppose now that $d<\dmax$. Let $\td{W}=\weight(\tdFopt)$; by assumption, $\td{W}\leq (9/10)^d\cdot W$.
Call an object $p\in \tdFopt$ {\em{heavy}} if 
$\weight(p)\geq s^{-2}\td{W}$ and let $\tdFheavy\subseteq \tdFopt$ be the family of heavy objects in $\tdFopt$. Observe that $|\tdFheavy|\leq s^2$,
so in the first loop of the algorithm there is an iteration in which we correctly fix the set $\tdFheavy$. We continue the reasoning under this assumption.
As in the description, let $\td{\Dd}'$ be $\td{\Dd}$ with all objects intersecting any object from $\tdFheavy$ removed.

Now let $\tdFopt'=\tdFopt\setminus \tdFheavy$. Since $\weight(p)<s^{-2}\td{W}$ for all $p\in \tdFopt'$, we may apply Lemma~\ref{lem:breaking} to $\tdFopt'$ and $\td{W}$ as the upper bound on its weight,
and conclude that the enumerated family $\Xfam$ contains at least one member $\Xx$ satisfying the properties as in the description of the algorithm.
In particular, in one iteration of the second loop the algorithm correctly fixes such $\Xx$ and proceeds with it. 
From now on we continue under this assumption. 

Having fixed $\tdFheavy$ and $\Xx$, the algorithm investigates the connected components $\td{\Dd}_1,\ldots,\td{\Dd}_k$ of $\IntGraph(\td{\Dd}')-\Xx$ and applies itself recursively to each instance $(G,\td{\Dd}_i)$,
for $i=1,\ldots,k$, yielding families $\td{\Ff}_i\subseteq \td{\Dd}_i$.
By the properties of $\Xx$, we have that 
$$\weight(\Fopt\cap\td{\Dd}_i)\leq (9/10)\cdot \weight(\Fopt\cap \td{\Dd})\leq (9/10)^{d+1}\cdot W,$$
for each $i\in \{1,\ldots,k\}$. Since recursive subcalls to instances $(G,\td{\Dd}_i)$ are at level $d+1$ in the recursion tree, by the induction assumption we infer that
\begin{equation}\label{eq:recursion}
\weight(\Ff_i)\geq (1-(d'-1)\epsloc)\cdot \weight(\Fopt\cap \td{\Dd}_i)
\end{equation}
for each $i\in \{1,\ldots,k\}$. Since the total weight of objects of $\tdFopt\cap \Xx$ is at most $\epsloc\td{W}$, we have that
\begin{equation}\label{eq:loss}
\weight(\tdFheavy)+\sum_{i=1}^k \weight(\tdFopt\cap \td{\Dd}_i)\geq (1-\epsloc)\td{W}.
\end{equation}
Recall that the algorithm returns a family of weight not smaller than the weight of $\td{\Ff}=\tdFheavy\cup \bigcup_{i=1}^k \td{\Ff}_i$.
By combining~\eqref{eq:recursion} with~\eqref{eq:loss} we have
\begin{eqnarray*}
\weight(\td{\Ff}) & =     & \weight(\tdFheavy)+\sum_{i=1}^k \weight(\td{\Ff}_i) \\
                  & \geq  & \weight(\tdFheavy)+(1-(d'-1)\epsloc)\cdot \sum_{i=1}^k \weight(\tdFopt\cap \td{\Dd}_i) \\
                  & \geq  & (1-(d'-1)\epsloc)\cdot \left(\weight(\tdFheavy)+\sum_{i=1}^k \weight(\tdFopt\cap \td{\Dd}_i)\right) \\
                  & \geq  & (1-(d'-1)\epsloc)(1-\epsloc)\td{W}\geq (1-d'\epsloc)\td{W}.
\end{eqnarray*}
This concludes the proof of Theorem \ref{thm:main-mwis}.
\cqed\end{proof}

\section{A QPTAS for {\sc{Minimum Weight Distance Set Cover}}}\label{sec:mwds}
In this section we prove Theorem~\ref{thm:main-mwds}. We first state and prove a suitable variant of Lemma~\ref{lem:breaking}; its proof will follow from Lemma~\ref{lem:balanced-sep} in a similar way.


\begin{lemma}[Separator Lemma for {\sc{MWDSC}}]\label{lem:breaking-ds}
Let $G$ be an $n$-vertex planar graph, $\Dd$ be a weighted set of vertices of $G$ with $|\Dd|=N$, $\Cc$ be a set of vertices of $G$, and $r$ be a nonnegative real.
Let $0<\eps<\frac{1}{20}$ and denote $s=10^3\cdot\frac{1}{\eps} \ln \frac{1}{\eps}$.
Then there is a family $\Xfam$ consisting of quadruples of the form $(\Dd_1,\Dd_2,\Cc_1,\Cc_2)$ such that:
\begin{enumerate}[label=(D\arabic*),ref=(D\arabic*)]
\item\label{p:osl2-union} for each $(\Dd_1,\Dd_2,\Cc_1,\Cc_2)\in \Xfam$, we have that $\Dd_1,\Dd_2\subseteq \Dd$, $\Cc_1,\Cc_2\subseteq \Cc$, and $\Cc_1\cup \Cc_2=\Cc$;
\item\label{p:osl2-size} $|\Xfam|\leq 6^{3s}N^{15s}$ and $\Xfam$ can be computed in time $N^{\Oh(s)}\cdot n^{\Oh(1)}$; and
\item\label{p:osl2-sep} for every real $W\geq 0$ and subset $\Ff\subseteq \Dd$ such that $\Ff$ $r$-covers $\Cc$, $\weight(\Ff)\leq W$, and $\weight(p)\leq s^{-2}W$ for each $p\in \Ff$,
there exists a quadruple $(\Dd_1,\Dd_2,\Cc_1,\Cc_2)\in \Xfam$ such that:
\begin{itemize}
\item $\weight(\Ff\cap \Dd_1\cap \Dd_2)\leq \eps W$;
\item $\weight(\Ff\cap \Dd_1)\leq \frac{19}{20}W$ and $\weight(\Ff\cap \Dd_2)\leq \frac{19}{20}W$; and
\item $\Ff\cap \Dd_1$ $r$-covers $\Cc_1$ and $\Ff\cap \Dd_2$ $r$-covers $\Cc_2$.
\end{itemize}
\end{enumerate}
\end{lemma}
\begin{proof}
As before, we may assume we are given an embedding of $G$ in a sphere $\Sigma$, $G$ is triangulated using edges of infinite weight, and distances in $G$ are unique.

Consider applying Lemma~\ref{lem:balanced-sep} to $\Dd$ (where we interpret vertices in $\Dd$ as single-vertex objects) and a subset $\Ff\subseteq \Dd$ as in \ref{p:osl2-sep}.
This yields a Voronoi separator $S$ of length at most $3s$ such that all faces traversed by $S$ are $\Dd$-important and $S$ breaks $\Ff$ in a balanced way as described in Lemma~\ref{lem:balanced-sep}.
As explained in the proof of Lemma~\ref{lem:breaking}, given only $\Dd$ there are at most $6^{3s}N^{15s}$ ways to choose such a Voronoi separator $S$ and 
in time $N^{\Oh(s)}\cdot n^{\Oh(1)}$ we can enumerate a family $\Vorsepfam$ of at most $6^{3s}N^{15s}$ candidates for $S$.

Now, construct $\Xfam$ by including the following quadruple $(\Dd_1,\Dd_2,\Cc_1,\Cc_2)$ for each $S\in \Vorsepfam$.
Let $\gamma=\pi(S)$ be the perimeter of $S$ (recall that the perimeter of $S$ is the noose ``following'' $S$; we defined it in Section~\ref{sec:nooses-to-seps}).
Then define $(\Dd_1,\Dd_2,\Cc_1,\Cc_2)$ as follows:
\begin{itemize}
\item $\Dd_1$ is the union of $\Ban(S)$ and all vertices of $\Dd$ contained in $\enc(\gamma)$;
\item $\Dd_2$ is the union of $\Ban(S)$ and all vertices of $\Dd$ contained in $\exc(\gamma)$;
\item $\Cc_1$ comprises all vertices of $\Cc$ contained in $\enc(\gamma)\cup \gamma$; and 
\item $\Cc_2$ comprises all vertices of $\Cc$ contained in $\exc(\gamma)\cup \gamma$.
\end{itemize}
This concludes the construction of $\Xfam$; we are left with verifying its properties.
Conditions \ref{p:osl2-union} and \ref{p:osl2-size} are clear from the construction, so we are left with condition \ref{p:osl2-sep}.
Suppose $W$ and $\Ff$ are as in the statement of condition \ref{p:osl2-sep}.

For the first assertion, observe that $\Dd_1\cap \Dd_2=\Ban(S)$ and by Lemma~\ref{lem:balanced-sep}, condition \ref{p:sacrifice}, we have that $\weight(\Ff\cap \Ban(S))\leq \eps W$. The claim follows.

For the second assertion, observe that $\Ff\cap \Dd_1$ comprises the vertices of $\Ff$ contained in $\enc(\gamma)$ and the vertices of $\Ff$ banned by $S$.
By Lemma~\ref{lem:balanced-sep}, conditions \ref{p:sacrifice} and \ref{p:balanced}, the weights of these sets of vertices are at most $\frac{9}{10}W$ and at most $\eps W\leq \frac{W}{20}$, respectively.
It follows that $\weight(\Ff\cap \Dd_1)\leq \frac{19}{20}W$; a symmetric reasoning shows that $\weight(\Ff\cap \Dd_2)\leq \frac{19}{20}W$ as~well.

We are left with the third assertion. Take any vertex $c\in \Cc_1$. Let $p$ be the vertex from $\Ff$ that is closest to $c$ and let $P$ be the shortest path connecting $c$ with $p$.
Since $\Ff$ $r$-covers $\Cc$, we have that $\mathrm{length}(P)=\dist(p,c)\leq r$.
If $p\in \Ff\cap \Dd_1$ then we are done, hence assume that $p\in \Ff\setminus \Dd_1$.
In particular, $p$ is contained in $\exc(\gamma)$ and $p$ is not banned by $S$.
Since $c\in \Cc_1$, $c$ is contained in $\enc(\gamma)\cup \gamma$, so $P$ has to cross $\gamma$ at some vertex, say $w$.
Let $q\in \Dd(S)$ be the vertex on the Voronoi separator $S$ that is the closest to $w$; recall that $\Dd(S)\subseteq \Ff$, hence $q\in \Ff$.
Since $p$ is not banned by $S$, we have that $\dist(w,p)>\dist(w,q)$. As $P$ is the shortest path connecting $c$ and $p$, we infer that
$$r\geq \dist(c,p)=\dist(c,w)+\dist(w,p)>\dist(c,w)+\dist(w,q)\geq \dist(c,q).$$
This means that $c$ is $r$-covered by $q$. But $q$ is banned by $S$ due to $q\in \Dd(S)\subseteq \Ban(S)$, which implies that $c$ is $r$-covered by $\Ff_1$.
Since $c$ was chosen arbitrarily, we infer that $\Cc_1$ is $r$-covered by $\Ff\cap \Dd_1$, and a symmetric reasoning shows that $\Cc_2$ is $r$-covered by $\Ff\cap \Dd_2$.
\end{proof}

\subparagraph*{Reduction of the weight span.}
Similarly as in Section~\ref{sec:mwis}, we first bound the weight span of the input instance, defined as
$$\spn(\Dd)=\frac{\max_{p\in \Dd}\weight(p)}{\min_{p\in \Dd}\weight(p)}.$$
by a polynomial of $|\Dd|$ and $1/\eps$.

\begin{claim}\label{cl:span-reduction-ds}
Suppose there is a QPTAS as described in Theorem~\ref{thm:main-mwds}, but working on instances $(G,\Dd,\Cc,r)$ that additionally satisfy $\spn(\Dd)\leq 2\eps^{-1}|\Dd|$.
Then Theorem~\ref{thm:main-mwds} holds.
\end{claim}
\begin{proof}
Let $M=2\eps^{-1}|\Dd|$.
For every vertex $p\in \Dd$ construct a weighted set of vertices $\Dd_p$ by modifying $\Dd$ as follows:
 remove all vertices of weight larger than $\weight(p)$ and for every vertex of weight smaller than $\weight(p)/M$, increase its weight to $\weight(p)/M$.
For clarity, the weight function in $\Dd_p$ will be denoted by $\weight_p(\cdot)$. That is, for each $q\in \Dd$ we have:
\begin{itemize}
\item if $\weight(p)<\weight(q)$, then $q\notin \Dd_p$;
\item if $\weight(p)/M<\weight(q)\leq \weight(p)$, then $q\in \Dd_p$ and $\weight_p(q)=\weight(q)$; and
\item if $\weight(q)\leq \weight(p)/M$, then $q\in \Dd_p$ and $\weight_p(q)=\weight(p)/M$.
\end{itemize}
Clearly $\spn(\Dd_p)\leq M$.
We claim that there exists $p\in \Dd$ such that the optimum solution for $(G,\Dd_p,\Cc,r)$ has weight at most $(1+\eps/2)$ larger than the optimum solution for $(G,\Dd,\Cc,r)$.
If this is the case, then it suffices to run the assumed QPTAS on each of the instances $(G,\Dd_p,\Cc,r)$, for $p\in \Dd$, using accuracy parameter $\eps/2$, 
and output the lightest of the solutions found.
Indeed, instances $(G,\Dd_p,\Cc,r)$ are derived from $(G,\Dd,\Cc,r)$ only by dropping some vertices and increasing the weights of some others, so every solution for $(G,\Dd_p,\Cc,r)$
naturally projects to a solution for $(G,\Dd,\Cc,r)$ of not larger weight.

Let $\Fopt\subseteq \Dd$ be the optimum solution for $(G,\Dd,\Cc,r)$. Let $p$ be the heaviest vertex in $\Fopt$.
Observe that $\Fopt$ is still a solution in $(G,\Dd_p,\Cc,r)$, yet its weight with respect to $\weight_p(\cdot)$ may be larger than with respect to $\weight(\cdot)$.
However, for each vertex $q\in \Fopt$ we have
$$\weight_p(q)\leq \weight(q)+\weight(p)/M.$$
Thus,
\begin{eqnarray*}
\weight_p(\Fopt) & \leq & \weight(\Fopt)+|\Fopt|\cdot \weight(p)/M\\
& = & \weight(\Fopt)+\eps/2\cdot \weight(p)\cdot \frac{|\Fopt|}{|\Dd|}\leq (1+\eps/2)\cdot \weight(\Fopt),
\end{eqnarray*}
where the last inequality follows from $p\in \Fopt$.
\cqed\end{proof}

Denoting $M=2\eps^{-1}N=2\eps^{-1}|\Dd|$, by Claim~\ref{cl:span-reduction-ds} in order to prove Theorem~\ref{thm:main-mwds} it suffices to work under the assumption that $\spn(\Dd)<M$.
By rescaling the weights from now on we assume that all the weights of vertices in $\Dd$ are at least $1$ and smaller than $M$.

\subparagraph*{Parameters.} Similarly as in Section~\ref{sec:mwis}, we set parameters as follows:
\begin{equation}\label{eq:params}
\dmax=20\ln(MN),\qquad \qquad \epsloc = \frac{\eps}{2\dmax},\qquad \qquad s=10^3\cdot\frac{1}{\epsloc} \ln \frac{1}{\epsloc}.
\end{equation}
Note that compared to the proof of Theorem~\ref{thm:main-mwis}, the recursion depth $\dmax$ is twice larger and the rescaled accuracy parameter $\epsloc$ is four times smaller.
This is for technical reasons.

\subparagraph*{Description of the algorithm.} We perform a similar recursive scheme as in Section~\ref{sec:mwis}, but using Lemma~\ref{lem:breaking-ds} instead of Lemma~\ref{lem:breaking}.
The algorithm is summarized using pseudocode as Algorithm~\ref{alg:main-ds}. We now explain it viewing it as a recursive search for the optimum solution.

Recall that we are given an instance $(G,\Dd,\Cc,r)$ and an accuracy parameter $\eps$, and we have fixed parameters $\dmax,\epsloc,s$ as in~\eqref{eq:params}.
Fix an optimum solution $\Fopt$ and denote $W=\weight(\Fopt)$.

We first guess heavy vertices in the solution as in the proof of Theorem~\ref{thm:main-mwis}.
Call a vertex $p\in \Fopt$ {\em{heavy}} if $\weight(p)>s^{-2}W$. Observe that $\Fopt$ contains at most $s^2$ heavy vertices; denote them by $\Fheavy$.
The algorithm iterates through all subsets of $\Dd$ of size at most $s^2$, in each case fixing the considered subset as $\Fheavy$.
In one of the cases $\Fheavy$ will be fixed correctly, hence we may proceed with the assumption that the algorithm knows $\Fheavy$.

Let $\Fopt'=\Fopt\setminus \Fheavy$ and $\Dd'=\Dd\setminus \Fheavy$.
Further, let $\Cc'$ be constructed from $\Cc$ by removing all vertices that are $r$-covered by $\Fheavy$.
Note that $\Cc'$ is $r$-covered by $\Fopt'$. 
Moreover, observe that $\weight(p)\leq s^{-2}W$ for each $p\in \Fopt'$, hence we may apply Lemma~\ref{lem:breaking-ds} to $\Dd'$ and $\Cc'$.
Thus, in time $N^{\Oh(s)}\cdot n^{\Oh(1)}$ we construct a family $\Xfam$ consisting of at most $6^{3s}N^{15s}$ quadruples with the following guarantee:
there exists a quadruple $(\Dd_1,\Dd_2,\Cc_1,\Cc_2)\in \Xfam$ such that
\begin{itemize}
\item $\Dd_1,\Dd_2\subseteq \Dd'$, $\Cc_1,\Cc_2\subseteq \Cc'$, and $\Cc_1\cup \Cc_2=\Cc'$;
\item $\weight(\Fopt'\cap \Dd_1\cap \Dd_2)\leq \epsloc W$;
\item $\weight(\Fopt'\cap \Dd_1)\leq \frac{19}{20}W$ and $\weight(\Fopt'\cap \Dd_2)\leq \frac{19}{20}W$; and
\item $\Fopt'\cap \Dd_1$ $r$-covers $\Cc_1$ and $\Fopt'\cap \Dd_2$ $r$-covers $\Cc_2$.
\end{itemize}
By iterating through all members of $\Xfam$, we may henceforth assume that the algorithm has fixed a quadruple $(\Dd_1,\Dd_2,\Cc_1,\Cc_2)$ with the properties as above.

The algorithm now recurses on instances $(G,\Dd_1,\Cc_1,r)$ and $(G,\Dd_2,\Cc_2,r)$, thus computing two solutions $\Ff_1\subseteq \Dd_1$ and $\Ff_2\subseteq \Dd_2$
such that $\Cc_1$ is $r$-covered by $\Ff_1$ and $\Cc_2$ is $r$-covered by $\Ff_2$. Since $\Cc_1\cup \Cc_2=\Cc'$ and all members of $\Cc\setminus \Cc'$ are $r$-covered by $\Fheavy$,
the set $\Ff=\Fheavy\cup \Ff_1\cup \Ff_2$ $r$-covers $\Cc$. The algorithm outputs the lightest of the sets $\Ff$ computed as above for all choices of $\Fheavy$ and $(\Dd_1,\Dd_2,\Cc_1,\Cc_2)\in \Xfam$.

The base case of the recursion is given by trimming it at depth $\dmax$. More precisely, all recursive calls at depth larger than $\dmax$ return $\Ff=\emptyset$ in the case when $\Cc$ is empty 
(then it is an optimal solution), or a special marker $\bot$ in the case when $\Cc$ is not empty. The marker $\bot$ should be interpreted as ``error'', that is, 
it marks that a recursive call has failed to find a solution, and it propagates in the recursion using the following convention: $\weight(\bot)=+\infty$ and the union of $\bot$ with any other set is equal to $\bot$.

\newcommand{\AlgNameDS}{\mathtt{AlgMWDSC}}

\begin{algorithm}[h!]
  \KwIn{An instance $(G,\Dd,\Cc,r)$, recursion depth $d$} 
  \KwOut{A subset $\Ff\subseteq \Dd$ that $r$-covers $\Cc$}  \Indp \BlankLine
  \If{$\Cc=\emptyset$}{\KwRet{$\emptyset$}}
  \If{$d>\dmax$}{\KwRet{$\bot$}}
  $\Ff\leftarrow \bot$\\
  \ForAll{\normalfont{$\Fheavy$: subset of $\Dd$ with $|\Fheavy|\leq s^2$}}{
        $\Dd'\leftarrow \Dd\setminus \Fheavy$\\
        $\Cc'\leftarrow \Cc\setminus (\textrm{vertices of }\Cc\textrm{ that are }r\textrm{-covered by }\Fheavy)$\\
        $\Xfam\leftarrow$ family computed for $\Dd'$ and $\Cc'$ using Lemma~\ref{lem:breaking-ds}\\
        \ForAll{$(\Dd_1,\Dd_2,\Cc_1,\Cc_2)\in \Xfam$}{
	  $\Ff_1 \leftarrow \AlgNameDS(G,\Dd_1,\Cc_1,r)$\\
	  $\Ff_2 \leftarrow \AlgNameDS(G,\Dd_2,\Cc_2,r)$\\
	  $\Ff_{\textrm{cand}}\leftarrow \Fheavy\cup \Ff_1\cup \Ff_2$\\
	  \If{$\weight(\Ff_{\mathrm{cand}})<\weight(\Ff)$}{$\Ff\leftarrow \Ff_{\textrm{cand}}$}
        }
   }
   \KwRet{$\Ff$}  
\caption{Algorithm $\AlgNameDS$}
  \label{alg:main-ds}
\end{algorithm}

Again, we need to argue that the running time of the algorithm is as promised and that the output solution has weight at most $(1+\eps)W$.
The running time analysis is exactly the same as in the proof of Theorem~\ref{thm:main-mwis}, so we skip it and proceed directly to arguing the approximation factor.

\subparagraph*{Approximation factor analysis.} Similarly as in the case of Theorem~\ref{thm:main-mwis},
the fact that the algorithm outputs a solution of weight at most $(1+\eps)W$ follows directly from the following claim by applying it to the original instance $(G,\Dd,\Cc,r)$ and $d=0$.
Recall here that $\Fopt$ is a fixed optimum solution to the original instance and its weight is $W$.

\begin{claim}
Let $(G,\td{\Dd},\td{\Cc},r)$ be an instance on which the algorithm is called in the recursion tree, say at depth $d$.
Suppose further that $\weight(\Fopt\cap \td{\Dd})\leq (19/20)^d\cdot W$ and that $\Fopt\cap \td{\Dd}$ $r$-covers $\td{\Cc}$.
Then the call of the algorithm to $(G,\td{\Dd},\td{\Cc},r)$ returns a subset $\td{\Ff}\subseteq \td{\Dd}$ that $r$-covers $\td{\Cc}$ and satisfies
$\weight(\td{\Ff})\leq (1+2d'\epsloc)\weight(\Fopt\cap \td{\Dd})$, where $d'=\dmax-d$.
\end{claim}
\begin{proof}
We prove the claim by induction on $d$, starting with the case $d\geq \dmax$ and then proceeding with decreasing $d$.
When $d\geq \dmax$, we have
$$(19/20)^d\leq (1-1/20)^{20\ln(MN)}<e^{-\ln(MN)}=\frac{1}{MN}.$$
Since $\Fopt$ consists of at most $N$ vertices of weight less than $M$ each, we have $W<MN$.
Therefore 
$$\weight(\Fopt\cap \td{\Dd})\leq (19/20)^d\cdot W<\frac{1}{MN}\cdot MN=1.$$
As each vertex of $\Dd$ has weight at least $1$, we infer that $\Fopt\cap \td{\Dd}=\emptyset$. This means that $\td{\Cc}$ is empty as well and the algorithm correctly outputs $\emptyset$ as a solution.

Suppose now that $d<\dmax$. Denote $\tdFopt=\Fopt\cap \td{D}$ and 
let $\td{W}=\weight(\tdFopt)$; by assumption, $\td{W}\leq (19/20)^d\cdot W$.
Call a vertex $p\in \tdFopt$ {\em{heavy}} if 
$\weight(p)\geq s^{-2}\td{W}$ and let $\tdFheavy\subseteq \tdFopt$ be the set of heavy vertices in $\tdFopt$. Observe that $|\tdFheavy|\leq s^2$,
so in the first loop of the algorithm there is an iteration in which we correctly fix the set $\tdFheavy$. We continue the reasoning under this assumption.
As in the description, let $\td{\Dd}'$ be $\td{\Dd}$ with all vertices from $\tdFheavy$ removed,
and let $\td{\Cc}'$ be $\td{\Cc}$ with all vertices $r$-covered by $\tdFheavy$ removed.

Now let $\tdFopt'=\tdFopt\setminus \tdFheavy$ and observe that $\tdFopt'$ $r$-covers $\Cc'$. 
Since $\weight(p)<s^{-2}\td{W}$ for all $p\in \tdFopt'$, we may apply Lemma~\ref{lem:breaking-ds} to $\tdFopt'$, $\Cc'$,
and $\td{W}$ as the upper bound on the weight of $\tdFopt'$.
Thus we conclude that the enumerated family $\Xfam$ contains at least one quadruple $(\td{\Dd}_1,\td{\Dd}_2,\td{\Cc}_1,\td{\Cc}_2)$ 
satisfying the following:
\begin{enumerate}[label=(\roman*),ref=(\roman*)]
\item\label{p:sum} $\weight(\tdFopt\cap \td{\Dd}_1\cap \td{\Dd}_2)\leq \epsloc\td{W}$;
\item\label{p:pro} $\weight(\tdFopt\cap \td{\Dd}_1)\leq \frac{19}{20}\td{W}$ and $\weight(\tdFopt\cap \td{\Dd}_2)\leq \frac{19}{20}\td{W}$; and
\item\label{p:cov} $\tdFopt\cap \td{\Dd}_1$ $r$-covers $\Cc_1$ and $\tdFopt\cap \td{\Dd}_2$ $r$-covers $\Cc_2$.
\end{enumerate}
In particular, in one iteration of the second loop the algorithm correctly fixes such a quadruple and proceeds with it. 
From now on we continue under this assumption. 

Having fixed $\tdFheavy$ and $(\td{\Dd}_1,\td{\Dd}_2,\td{\Cc}_1,\td{\Cc}_2)$, 
the algorithm recursively calls itself on instances $(G,\td{\Dd}_1,\td{\Cc}_1,r)$ and $(G,\td{\Dd}_2,\td{\Cc}_2,r)$,
yielding families $\td{\Ff}_1\subseteq \td{\Dd}_1$ and $\td{\Ff}_2\subseteq \td{\Dd}_2$.
By property~\ref{p:pro}, we have that 
$$\weight(\Fopt\cap\td{\Dd}_i)\leq (19/20)\cdot \weight(\Fopt\cap \td{\Dd})\leq (19/20)^{d+1}\cdot W,$$
for $i\in \{1,2\}$. As the recursive subcalls to instances $(G,\td{\Dd}_i,\td{\Cc}_i,r)$ are at level $d+1$ in the recursion tree, and due to property~\ref{p:cov},
by the induction assumption we infer that $\td{\Ff}_i\neq \bot$ and
\begin{equation}\label{eq:recursion-ds}
\weight(\td{\Ff}_i)\leq (1+2(d'-1)\epsloc)\cdot \weight(\Fopt\cap \td{\Dd}_i),
\end{equation}
for $i\in \{1,2\}$. 
By property~\ref{p:sum}, we have that
\begin{eqnarray}\label{eq:loss-ds}
&   & \weight(\tdFheavy)+\weight(\tdFopt\cap \td{\Dd}_1)+\weight(\tdFopt\cap \td{\Dd}_2)\nonumber \\
& = & \weight(\tdFheavy)+\weight(\tdFopt\cap \td{\Dd}')+\weight(\tdFopt\cap \td{\Dd}_1\cap \td{\Dd}_2)\nonumber \\
& = & \weight(\tdFopt)+\weight(\tdFopt\cap \td{\Dd}_1\cap \td{\Dd}_2)\nonumber\\
& \leq & \td{W}+\epsloc\td{W}= (1+\epsloc)\td{W}.
\end{eqnarray}
Recall that the algorithm returns a solution of weight not larger than the weight of $\td{\Ff}=\tdFheavy\cup \td{\Ff}_1\cup \td{\Ff}_2$.
By combining~\eqref{eq:recursion-ds} with~\eqref{eq:loss-ds} we have
\begin{eqnarray*}
\weight(\td{\Ff}) & \leq  & \weight(\tdFheavy)+\weight(\td{\Ff}_1)+\weight(\td{\Ff}_2) \\
                  & \leq  & \weight(\tdFheavy)+(1+2(d'-1)\epsloc)\cdot\left(\weight(\Fopt\cap \td{\Dd}_1)+\weight(\Fopt\cap \td{\Dd}_2)\right) \\
                  & \leq  & (1+2(d'-1)\epsloc)\cdot \left(\weight(\tdFheavy)+\weight(\tdFopt\cap \td{\Dd}_1)+\weight(\tdFopt\cap \td{\Dd}_2)\right) \\
                  & \leq  & (1+2(d'-1)\epsloc)(1+\epsloc)\td{W}\leq (1+2d'\epsloc)\td{W},
\end{eqnarray*}
where the last inequality follows from $2(d'-1)\epsloc^2\leq \epsloc$, which is true for $d'\leq \dmax\leq \frac{1}{40\epsloc}$.
This concludes the proof.
\cqed\end{proof}

This finishes the proof of Theorem~\ref{thm:main-mwds}.

\subparagraph*{Acknowledgements.} We thank D\'aniel Marx for insightful discussions on the approach to optimization problems in geometric and planar settings via Voronoi diagrams.

\bibliographystyle{plainurl}
\bibliography{citations}

\appendix

\section{Geometric problems}\label{app:geometric}
\newcommand{\crGraph}[3]{#1\langle #2,#3\rangle}

In the {\sc{Maximum Weight Independent Set of Polygons}} ({\sc{MWISP}}) problems we are given a family $\Pp$ of polygons in the plane, each with a prescribed weight, 
and the task is to find a maximum-weight subset of pairwise disjoint polygons.
In the {\sc{Weighted Geometric Set Cover}} problem we are given a family $\Pp$ of subsets of the plane,
each with a prescribed weight, and a set $\Uu$ of points in the plane. The goal is to find a minimum-weight subfamily of $\Pp$ whose union covers $\Uu$.
A QPTAS for {\sc{MWISP}} with running time $2^{\poly(1/\eps,\log N)}\cdot n^{\Oh(1)}$, where $N=|\Pp|$ and $n$ is the total number of vertices of the input polygons, was given by Har-Peled~\cite{Har-Peled2014}.
A QPTAS for {\sc{WGSR}} under the assumption that $\Pp$ is a family of {\em{pseudo-disks}} (i.e. compact, simply connected subsets of the plane such that
the boundaries of each two meet in at most two points) with running time $2^{\poly(1/\eps,n+N)}$, where $n=|\Uu|$ and $N=|\Pp|$, was given by Mustafa et al.~\cite{MustafaRR15}.

We now prove that the abovementioned results follow from Theorems~\ref{thm:main-mwis} and~\ref{thm:main-mwds} via simple reductions from the geometric to the planar setting;
the exception is that for {\sc{WGSR}} we can tackle only unit disks and unit squares instead of general pseudo-disks.
These reductions were already observed by Marx and Pilipczuk in~\cite{MarxP15}, who used them for the parameterized variants of the problems. We give them for completeness.

\begin{corollary}\label{cor:mwisp}
The {\sc{Maximum Weight Independent Set of Polygons}} problem admits a QPTAS with running time $2^{\poly(1/\eps,\log N)}\cdot n^{\Oh(1)}$, 
where $N$ is the number of polygons and $n$ is the total number of vertices of those polygons.
\end{corollary}

\begin{corollary}\label{cor:geom-setcover}
The {\sc{Weighted Geometric Set Cover}} problems for unit disks and axis-parallel unit squares admits a QPTAS with running time $2^{\poly(1/\eps,\log N)}\cdot n^{\Oh(1)}$, 
where $N$ is the number of disks/squares on input and $n$ is the number of points to be covered.
\end{corollary}

Before we proceed to proving the statements above, let us first give a general-use definition of the {\em{crossing graph}} of a point set in the plane; see Figure~\ref{fig:crgraph}.
Consider a finite set of points $X$ in the plane and a metric $d(\cdot,\cdot)$ on the plane induced by some norm (i.e. $d(x,y)=\|x-y\|$ for some norm $\|\cdot\|$).
By $\crGraph{G}{X}{d}$ we denote a weighted graph defined as follows.
For every pair of distinct points $a,b\in X$, draw the segment with endpoints $a$ and $b$ in the plane. Let $Y\supseteq X$ be the set consisting of $X$ and all intersections of all the segments drawn;
then set $Y$ to be the vertex set of $\crGraph{G}{X}{d}$. Two points $x,y\in Y$ are connected by an edge in $\crGraph{G}{X}{d}$
if they are two consecutive points from $Y$ on any of the drawn segments; that is, the segment connecting
$x$ and $y$ is contained in one of the segments connecting pairs of vertices from $X$, and there is no other point from $Y$ inside this segment.
The weight of this edge is set to $d(x,y)$.
Note that $\crGraph{G}{X}{d}$ has at most $|X|^4$ vertices, is planar with a straight-line embedding described above, and for every two vertices $a,b\in X$ we have $\dist_{\crGraph{G}{X}{d}}(a,b)=d(a,b)$.

\begin{figure}[h]
                \centering
                \def\svgwidth{0.8\columnwidth}
                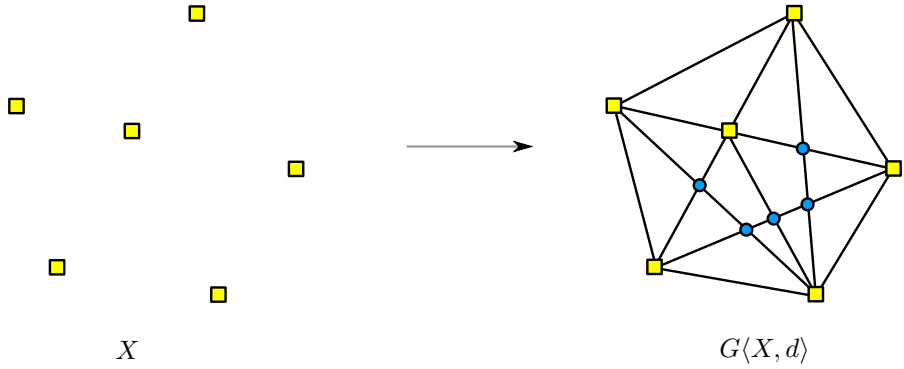
\caption{A point set $X$ and the corresponding crossing graph $\crGraph{G}{X}{d}$, for an unspecified metric~$d$. Note that two triples of points in $X$ are collinear.}\label{fig:crgraph}
\end{figure}

For further reference, let $d_2(\cdot,\cdot)$ and $d_{\infty}(\cdot,\cdot)$ denote the metrics on the plane induced by the $\ell_2$- and $\ell_\infty$-norm, respectively.
We now prove Corollaries~\ref{cor:mwisp} and~\ref{cor:geom-setcover}.

\begin{proof}[Proof of Corollary~\ref{cor:mwisp}]
Let $\cal P$ be the given weighted set of polygons and let $X$ be the set of all their vertices; then $|X|=n$.
Consider the graph $G=\crGraph{G}{X}{d_\infty}$; note that $G$ has at most $n^4$ vertices.
For every polygon $P\in \cal P$ construct a corresponding object $p$ in $G$ defined as follows: $p$ is the subgraph of $G$ induced by all the vertices of $G$ whose embeddings are contained in $P$.
It is easy to see that $p$ defined in this manner is connected. The weight of $p$ is equal to the weight of $P$. Let $\Dd$ be the family comprising all the constructed objects.
It readily follows that sets of pairwise disjoint polygons from $\cal P$ are in one-to-one correspondence with independent subfamilies of $\Dd$, hence it suffices to apply the algorithm of Theorem~\ref{thm:main-mwis}
to the instance $(G,\Dd)$.
\end{proof}

\begin{proof}[Proof of Corollary~\ref{cor:geom-setcover}]
Let us first concentrate on the case of unit disks. Let the input be $(\Pp,\Uu)$, where $\Pp$ is a weighted set of unit disks and $\Uu$ is a set of points to be covered.
Let $X$ be the set consisting of all the centers of disks from $\Pp$ and all the points from $\Uu$; then $|X|\leq N+n$.
Consider the graph $G=\crGraph{G}{X}{d_2}$; note that $G$ has at most $(N+n)^4$ vertices.
Construct a weighted set of points $\Dd$ as follows: for every disk $P\in \cal P$, add the center of $P$ to $\Dd$ and assign it weight equal to the weight of $P$.
It can be easily seen that a subset of disks ${\cal Q}\subseteq \Pp$ covers all the points from $\Uu$ if and only if the centers of disks from ${\cal Q}$ cover $\Uu$ in $G$, 
where we consider domination radius $\frac{1}{2}$. Hence it suffices to apply the algorithm of Theorem~\ref{thm:main-mwds} to the instance $(G,\Dd,\Uu,\frac{1}{2})$.

For the case of unit squares we may follow the same reasoning except we use metric $d_\infty(\cdot,\cdot)$ instead of $d_2(\cdot,\cdot)$.
\end{proof}

\end{document}